\documentclass[a4paper,UKenglish,cleveref, autoref, thm-restate]{lipics-v2021}

\usepackage{url}
\usepackage{hyperref}

\usepackage{microtype}
\usepackage[utf8]{inputenc}
\usepackage{amsmath, amsthm, amssymb, mathtools, stmaryrd}
\usepackage{xspace}
\usepackage{todonotes}
\usepackage{enumerate}
\usepackage{listings}
\usepackage{algorithm}
\usepackage{comment}
\usepackage{stmaryrd}

\usetikzlibrary{backgrounds}

\usepackage{algorithmic}

\usepackage{amsfonts}
\usepackage{color}
\usepackage{extarrows}
\usepackage{caption}
\usepackage{subcaption}
\usepackage{tabularx}
\usepackage[figuresright]{rotating}

\usepackage{tikz}

\hideLIPIcs
\nolinenumbers

\makeatletter
\newcommand{\@abbrev}[3]{
	\def\c@a@def##1{
		\if ##1.
		\relax
		\else
		\@ifdefinable{\@nameuse{#1##1}}{\@namedef{#1##1}{#2##1}}
		\expandafter\c@a@def
		\fi
	}
	\c@a@def #3.
}
\@abbrev{bb}{\mathbb}{ABCDEFGHIJKLMNOPQRSTUVWXYZ}
\@abbrev{bf}{\mathbf}{ABCDEFGHIJKLMNOPQRSTUVWXYZabcdefghijklmnopqrstuvwxyz}
\@abbrev{bit}{\boldsymbol}{
	ABCDEFGHIJKLMNOPQRSTUVWXYZabcdefghijklmnopqrstuvwxyz}
\@abbrev{mc}{\mathcal}{ABCDEFGHIJKLMNOPQRSTUVWXYZ}
\@abbrev{mb}{\mathbb}{ABCDEFGHIJKLMNOPQRSTUVWXYZ}
\@abbrev{mf}{\mathfrak}{
	ABCDEFGHIJKLMNOPQRSTUVWXYZabcdefghijklmnopqrstuvwxyz}
\@abbrev{rm}{\mathrm}{ABCDEFGHIJKLMNOPQRSTUVWXYZabcdefghijklmnopqrstuvwxyz}
\@abbrev{scr}{\mathscr}{
	ABCDEFGHIJKLMNOPQRSTUVWXYZabcdefghijklmnopqrstuvwxyz}
\@abbrev{sf}{\mathsf}{ABCDEFGHIJKLMNOPQRSTUVWXYZabcdefghijklmnopqrstuvwxyz}
\@abbrev{b}{\bar}{abcdefghijklmnopqrstuvwxyz}
\makeatother

\renewcommand{\phi}{\varphi}

\newcommand{\lra}{\longrightarrow}
\renewcommand{\phi}{\varphi}
\renewcommand{\theta}{\vartheta}

\renewcommand{\AA}{{\mathfrak A}}
\newcommand{\BB}{{\mathfrak B}}

\newcommand{\KK}{{\mathfrak K}}

\renewcommand{\epsilon}{\varepsilon}

\newcommand{\Cc}{{\cal C}}

\newcommand{\Ee}{{\cal E}}
\newcommand{\Ff}{{\cal F}}

\newcommand{\Kk}{{\cal K}}
\newcommand{\Ll}{{\cal L}}
\newcommand{\Mm}{{\cal M}}

\newcommand{\Oo}{{\cal O}}
\newcommand{\Pp}{{\cal P}}
\newcommand{\Qq}{{\cal Q}}

\newcommand{\Vv}{{\cal V}}

\renewcommand{\bar}{\overline}

\newcommand{\LA}[2]{\text{\upshape LA}^{#1}(#2)}
\newcommand{\CFI}{\mathsf{CFI}}
\newcommand{\IM}{\mathsf{IM}}
\newcommand{\C}[1]{\Cc^{#1}}
\newcommand{\Cmod}[1]{\Cc[#1]}
\newcommand{\CkMod}[2]{\Cc^{#1}[#2]}

\newcommand{\homEquiv}[1]{\equiv_{#1}}
\newcommand{\ImEquiv}[2]{\equiv^{\text{\upshape IM}}_{#1,#2}}

\newcommand{\homEquivMod}[2]{\equiv_{#1}^{#2}}
\newcommand{\equivC}[1]{\equiv_{\C{#1}}}

\DeclareMathOperator{\Hom}{Hom}
\DeclareMathOperator{\id}{id}

\DeclarePairedDelimiterX\setcond[2]{\{}{\}}{\mathchoice{\,}{}{}{}#1 \;\delimsize\vert\; #2\mathchoice{\,}{}{}{}}

\title{Limitations of Game Comonads via Homomorphism Indistinguishability}

\author{Moritz Lichter}{RWTH Aachen University, Germany}{lichter@lics.rwth-aachen.de}{https://orcid.org/0000-0001-5437-8074}{European Research Council (ERC) under the European Union’s Horizon 2020 research and innovation programme (EngageS: agreement No.\ 820148)}
\author{Benedikt Pago}{RWTH Aachen University, Germany }{pago@logic.rwth-aachen.de}{https://orcid.org/0000-0001-6377-1230}{}\author{Tim Seppelt}{RWTH Aachen University, Germany }{seppelt@cs.rwth-aachen.de}{https://orcid.org/0000-0002-6447-0568}{German Research Council (DFG) within Research Training Group 2236 (UnRAVeL)}


\authorrunning{M. Lichter, B. Pago, T. Seppelt} 

\Copyright{Moritz Lichter, Benedikt Pago, Tim Seppelt}

\begin{document}
	
	\maketitle

\begin{abstract}
Abramsky, Dawar, and Wang (2017) introduced the pebbling comonad for $k$-variable counting logic and thereby initiated a line of work that imports category theoretic machinery to finite model theory. 
Such game comonads have been developed for various logics, yielding characterisations of logical equivalences in terms of isomorphisms in the associated co-Kleisli category. 
We show a first limitation of this approach by studying linear-algebraic logic, which is strictly more expressive than first-order counting logic and whose $k$-variable logical equivalence relations are known as invertible-map equivalences (IM).
We show that there exists no finite-rank comonad on the category of graphs whose co-Kleisli isomorphisms characterise IM-equivalence, answering a question of Ó~Conghaile and Dawar (CSL 2021).
We obtain this result by ruling out a characterisation of IM-equivalence in terms of homomorphism indistinguishability and employing the Lovász-type theorems for game comonads established by Dawar, Jakl, and Reggio (2021).
Two graphs are homomorphism indistinguishable over a graph class
if they admit the same number of homomorphisms from every graph in the class. The IM-equivalences cannot be characterised in this way, neither when counting homomorphisms in the natural numbers, nor in any finite prime field.
\end{abstract}

\section{Introduction}
Logic fragments such as $k$-variable first-order logic with or without counting quantifiers induce equivalence relations on graphs, or more generally, on structures: Two structures are equivalent in this sense if they satisfy exactly the same sentences of the respective logic fragment. Such equivalence relations are approximations of the isomorphism relation. The more expressive the logic fragment, the more non-isomorphic structures are distinguished by it.
Classical model-comparison games and counterexamples like the Cai--Fürer--Immerman (CFI) construction show that $k$-variable FO (even with counting) does not distinguish all pairs of non-isomorphic structures. Hence, the induced equivalence is indeed strictly coarser than isomorphism.
Such approximations of isomorphism can be studied from many different angles. For example, it is well-known that counting logic equivalence is the same as indistinguishability by the Weisfeiler--Leman graph isomorphism test \cite{caifurimm92}. 

Another perspective to approximations of isomorphism is offered by homomorphism indistinguishability:
Two graphs $G$ and $H$ are \emph{homomorphism indistinguishable} over a class of graphs $\mathcal{F}$ if for all $F \in \mathcal{F}$ the number of homomorphisms from $F$ to $G$ is equal to the number of homomorphisms from $F$ to $H$.
Equivalence relations with respect to many logic fragments can be characterised as homomorphism indistinguishability relations over some graph class.
For example, two graphs are counting logic equivalent if and only if they are homomorphism indistinguishable over all graphs of bounded treewidth \cite{dvorak_recognizing_2010,dell_lovasz_2018}.
Besides counting logic equivalence, many other natural equivalence relations between graphs, including isomorphism~\cite{lovasz_operations_1967}, quantum isomorphism~\cite{mancinska_quantum_2020}, cospectrality~\cite{dell_lovasz_2018}, and feasibility of integer programming relaxations for graph isomorphism~\cite{dell_lovasz_2018,grohe_homomorphism_2022,roberson_lasserre_2023} have been characterised as homomorphism indistinguishability relations over various graph classes.
Characterising (logical) equivalences as homomorphism indistinguishability relations is desirable because such characterisations allow to compare the expressive power of logics solely by comparing the graph classes from which homomorphisms are counted \cite{roberson_lasserre_2023,roberson2022oddomorphisms}.
In this way, deep results from structural graph theory are made available for studying the expressive power of logics~\cite{seppelt2023logical}.

It is natural to ask whether this approach can be extended to interesting logics that are more expressive than counting logic, as they are for example studied in the quest for a logic for \textsc{Ptime}.
Such examples are \emph{rank logic}~\cite{rankLogicDiscovery, GradelPakusa19} and the more general \emph{linear-algebraic logic} (LA)~\cite{limitationsOfIM}.
We answer this question in the negative.
The \emph{invertible-map equivalence} $\ImEquiv{k}{\mathbb{P}}$,
as the equivalence of the $k$-variable fragment of LA is called, cannot be characterised as a homomorphism indistinguishability relation.

\begin{theorem} \label{thm:main}
	For every $k\geq6$, $\ImEquiv{k}{\mathbb{P}}$
	is not a homomorphism indistinguishability relation.
\end{theorem}	
The proof relies on CFI-like constructions such as the one from \cite{lichter2023separating} which was used by the first author to separate rank logic from polynomial time. We combine this with results by Roberson~\cite{roberson2022oddomorphisms} in order to obtain graphs which are invertible-map equivalent but not quantum isomorphic. As shown by the third author \cite{seppelt2023logical}, this suffices to conclude that invertible-map equivalence is not a homomorphism indistinguishability relation --  if it were, then it would have to be a refinement of quantum isomorphism.

\Cref{thm:main} also implies a negative answer to a question raised by Ó~Conghaile and Dawar~\cite{conghaile_game_2021}. Their work is part of a recent line of research aiming to characterise logical equivalences via a notion from category theory, namely as \emph{co-Kleisli isomorphism} for certain \emph{game comonads}. Ó~Conghaile and Dawar asked whether such a game comonad can be constructed for linear-algebraic logic. Employing a categorical Lovász-type theorem   \cite{dawar_lovasz-type_2021} that allows to infer the existence of a homomorphism indistinguishability relation from the existence of appropriate game comonads, we obtain the following result. To our knowledge, this is the first provable limitation of such comonadic characterisations.
\begin{theorem}
	\label{thm:no-comonad}
	For every $k \geq 6$,
	there is no finite-rank comonad $\mathbb{C}$ on the category of graphs such that $\ImEquiv{k}{\mathbb{P}}$ coincides with the isomorphism relation in the co-Kleisli category of~$\mathbb{C}$.
\end{theorem}
In this context, the concept of a comonad is best explained by recalling the \emph{pebbling comonad}~$\mathbb{T}_k$ introduced by Abramsky, Dawar, and Wang~\cite{abramsky_pebbling_2017}. Designed to provide a categorical formulation of the $k$-pebble game from finite model theory, it can be thought of as map sending structures to structures encoding Spoiler's plays in this game. Being a comonad, it gives rise to a category, the co-Kleisli category, whose objects are graphs and whose morphisms can be interpreted as winning strategies for Duplicator in the $k$-pebble game. Various notions from finite model theory can now be recovered from this construction: For example, a graph has treewidth less than $k$ if and only if it admits a $\mathbb{T}_k$-coalgebra. Crucially, two graphs satisfy the same $k$-variable counting logic sentences if and only if they are isomorphic in the co-Kleisli category of~$\mathbb{T}_k$.
Subsequently, comonads for many fragments \cite{abramsky_pebbling_2017,abramsky_relating_2021,montacute_pebble-relation_2022} and extensions \cite{conghaile_game_2021} of first-order logic have been constructed. 
They have in common that their co-Kleisli morphisms and isomorphisms encode winning strategies for Duplicator in one-sided, symmetric, and bijective games. Our \cref{thm:no-comonad} rules out that invertible-map equivalence can be dealt with along similar lines.

Comonads on the category of graphs and homomorphism indistinguishability are intimately connected.
Every homomorphism indistinguishability relation over a graph class with mild closure properties can be characterised as co-Kleisli isomorphism over a comonad~\cite{abramsky_discrete_2022}.
Conversely, the existence of co-Kleisli isomorphisms over some comonad that sends finite structures to finite structures can be characterised as a homomorphism indistinguishability relation \cite{dawar_lovasz-type_2021,reggio_polyadic_2021}.
This fundamental connection between comonads and homomorphism counting relations is exactly the reason why we can conclude the impossibility of the former from the impossibility of the latter: There is no finite-rank comonad for linear-algebraic logic.

 Hence, linear-algebraic logic seems to be of a very different nature than the weaker counting logic as it does not connect with the theory revolving around homomorphism counting and game comonads. This raises the question as to what is the precise reason for this situation. What makes a logic ``nice enough'' to fit within the homomorphism indistinguishability and comonadic framework? 
We can at least say that the shortcomings of LA in this respect are not due to it being strictly stronger than counting logic. There does exist an extension of counting logic which admits a comonad construction and thereby a homomorphism indistinguishability relation: This is $k$-variable infinitary FO enriched with all possible $n$-ary \emph{generalised quantifiers} over one-dimensional interpretations~\cite{conghaile_game_2021}. An $n$-ary generalised quantifier (also known as \emph{Lindström quantifier}) is essentially a membership oracle for a class $\Kk$ (of at most $n$-ary structures) that allows to test whether some structure~$\BB$ interpretable in the given structure $\AA$ is in $\Kk$.
 LA lies somewhere between counting logic and its extension by \emph{all} binary Lindström quantifiers because LA is infinitary FO extended with a \emph{proper subclass} of binary Lindström quantifiers. As a side node, counting logic itself is nothing but the extension of FO with \emph{all unary} Lindström quantifiers~\cite{kolaitis1995generalized}. Hence, we can describe the situation as follows: Whenever a Lindström-extension of infinitary FO contains \emph{all} one-dimensional Lindström quantifiers up to a given arity $n$, then it admits a comonad. If it only contains a subset of these Lindström 
 quantifiers, then this is not necessarily the case (our \cref{thm:main} is true even when we restrict LA to one-dimensional interpretations).

 Finally, another direction that we explore in this paper is counting homomorphism  in \emph{finite prime fields}. A large part of the theory of homomorphism indistinguishability that has been established so far works over the natural numbers. Given the fact that the linear-algebraic operators in LA are over finite fields, one might a priori suspect that the appropriate homomorphism indistinguishability relation must be based on homomorphism counts modulo a prime. However, this can also be ruled out, even when the homomorphisms are counted modulo several primes (\cref{thm:IMnoHomRelationModP}).
 
 As a positive result concerning homomorphism counting modulo primes, we find that Dvo\v{r}\'ak's  proof \cite{dvorak_recognizing_2010} can be adapted to finite fields: Two graphs admit the same numbers of homomorphisms modulo $p$ from all graphs of treewidth less than $k$ if and only if they are equivalent with respect to $k$-variable FO with mod-$p$ counting quantifiers (\cref{thm:dvorak-mod-p}).
 
 \begin{comment}
 It remains as an open problem to develop the theory of homomorphism counting modulo primes more in-depth. For example, an interesting question, which is beyond the scope of this article, though, is what kind of approximation to isomorphism is obtained when two graphs are said to be equivalent if the standard equation system for graph isomorphism has a solution in a prime field.
\todo{maybe Tim can explain this more}
\todo{Ich glaube ich würde diesen Absatz rausnehmen.}
\todo{Ja stimmt, das sollten wir. Vielleicht kann man das auch im Schlusskapitel erwähnen als Ausblick}
 
\end{comment}

\section{Preliminaries}
All structures in this article are relational and finite. When we speak of graphs, we mean $\{E\}$-structures, where $E$ is binary. When nothing else is specified, graphs are undirected. General relational structures are usually denoted $\AA$ or $\BB$, with $A$ or $B$ being used for the universe. Structures that are graphs will be written as $G = (V,E)$.
The set $\{1,2,...,n\}$ is denoted as $[n]$, and $\bbP \subseteq \bbN$ denotes the set of primes.

\subparagraph*{Counting logic.} The logic $\C{k}$ is the $k$-variable fragment of first-order logic with counting quantifiers of the form $\exists^{\geq i} x$, for every $i \in \bbN$. The semantics is as expected, i.e., a structure~$\AA$ satisfies a sentence $\exists^{\geq i} x \phi(x)$ if there exist at least $i$ distinct $a \in A$ such that $\AA \models \phi(a)$. We write $\AA \equivC{k} \BB$ if $\AA$ and $\BB$ are $\C{k}$-equivalent, i.e., they satisfy exactly the same $\C{k}$-sentences.

\subparagraph*{Lindström quantifiers and interpretations.} A more general way to extend FO is with \emph{Lindström} quantifiers (also known as \emph{generalised} quantifiers). A Lindström quantifier is essentially a membership oracle for a class of structures.
Before introducing Lindström quantifiers, we need the concept of logical interpretations.
Let $\sigma, \tau$ be vocabularies with $\tau = \{R_1,...,R_m\}$ where each $R_i$ is a relation symbol of arity $r_i$, and let $\Ll$ be a logic. An
\emph{$\ell$-dimensional $\Ll[\sigma,\tau]$-interpretation} $I$ is an $\Ll$-definable mapping from $\sigma$\nobreakdash-structures to $\tau$\nobreakdash-structures.
The elements of the $\tau$-structure are sets of $\ell$-tuples in the original $\sigma$-structure.
Generally, interpretations can take a tuple of parameters $\bar{z}$:
An $\ell$-dimensional $\Ll$-interpretation (with parameters) is a tuple
\[
I(\bar{z}) = \big(\phi_\delta(\bar{x}, \bar{z}), \phi_\approx(\bar{x},\bar{y}, \bar{z}), \phi_{R_1}(\bar{x}_1,...,\bar{x}_{r_1},\bar{z}), ..., \phi_{R_m}(\bar{x}_1,...,\bar{x}_{r_m},\bar{z})\big),
\]
where $\bar{x}, \bar{y}, \bar{x}_i$ are $\ell$-tuples of variables, and 
$\phi_\delta,\phi_\approx, \phi_{R_i}$ are $\sigma$-formulas of $\Ll$. The interpretation $I(\bar{z})$ defines a partial mapping from $\sigma$-structures to $\tau$-structures. For a given $\sigma$-structure $\AA$ and an assignment $\bar{z} \mapsto \bar{a}$, we define $\BB$ be as a $\tau$-structure with universe $B := \setcond{ \bar{b} \in A^k }{ \AA \models \phi_\delta(\bar{b}, \bar{a}) }$ and relations $R_i^\BB := \setcond{  (\bar{b}_1,...,\bar{b}_{r_i}) \in B^{r_i}  }{ \AA \models \phi_{R_i}(\bar{b}_1,...,\bar{b}_{r_i},\bar{a})  }$, for all $i \in [m]$. From this structure, the ``output'' $I(\AA, \bar{z} \mapsto \bar{a})$ is obtained by factoring out the equivalence classes defined by $\phi_\approx$. Formally, let $\Ee := \setcond{ (\bar{b}_1, \bar{b}_2) \in A^{2k} }{ \AA \models \phi_{\approx}(\bar{b}_1, \bar{b}_2,\bar{a})  }$. If $\Ee$ is not a congruence relation on $\BB$, then $I(\AA, \bar{z} \mapsto \bar{a})$ is undefined.
Otherwise, $I(\AA, \bar{z} \mapsto \bar{a})$ is defined to be the quotient structure $\BB/\Ee$.

Let $\Kk$ be a class of $\tau$-structures and $\Ll$ be a logic. The extension $\Ll(\Qq_\Kk)$ of $\Ll$ by the \emph{Lindström quantifier} for $\Kk$ is obtained by closing $\Ll$ under the following formula formation rule: Whenever $I(\bar{x})$ is an $\Ll(\Qq_\Kk)[\sigma,\tau]$-interpretation, then $\Qq_\Kk I(\bar{x})$ is a $\tau$-formula of $\Ll(\Qq_\Kk)$ with free variables $\bar{x}$. For a $\tau$-structure $\AA$ and an assignment $\bar{x} \mapsto \bar{a}$, it holds $(\AA, \bar{a}) \models \Qq_\KK I(\bar{x})$ if $I(\AA, \bar{x} \mapsto \bar{a}) \in \Kk$. If $\mathbf{Q}$ is a class of Lindström quantifiers, then $\Ll(\mathbf{Q})$ denotes the extension by all Lindström quantifiers in $\mathbf{Q}$. When we speak of the one-dimensional restriction of such a logic, we mean that in formulas $\Qq_\Kk I(\bar{x})$, the interpretation $I$ has to be one-dimensional.

\subparagraph*{Linear-algebraic logic and invertible-map equivalences.} Linear-algebraic logic (LA) was introduced by Dawar, Grädel, and Pakusa~\cite{limitationsOfIM} as an extension of infinitary first-order logic with all isomorphism-invariant linear-algebraic operators. As such, it extends \emph{rank logic} \cite{rankLogicDiscovery,GradelPakusa19}. Rank logic in turn is an extension of FO with operators for determining the rank of a matrix that is definable in the input structure. In linear-algebraic logic, formulas have access to \emph{any} isomorphism-invariant parameter of a definable matrix, not only the rank. The motivation for studying such a logic was to show that no linear-algebraic operators whatsoever can enhance the power of FO such that its $k$-variable fragment distinguishes all non-isomorphic structures, for some fixed $k$. For the detailed definition of LA, we refer to \cite{limitationsOfIM}; in short, LA is the closure of FO under infinite conjunctions and disjunctions and under all Lindström quantifiers of the form $\Qq_f^{t,\ell}I(\bar{x})$. Here, $I$ is an $\ell$-dimensional interpretation and $f$ is any \emph{linear-algebraic} function over some field $\bbF$ with some arity $m  \geq 1$ that maps tuples $(M_1,...,M_m)$ of linear transformations/matrices over $\bbF$ to natural numbers. For instance, the rank operator is such a function with $m=1$ that maps a given matrix to its rank. The condition that $f$ is \emph{linear-algebraic} means that $f$ is invariant under vector space isomorphisms.
Formally, this means that whenever two sequences of matrices $M_1,...,M_m$ and $M'_1,...,M'_m$ over $\bbF$ are \emph{simultaneously similar}, then $f(M_1,...,M_m) = f(M'_1,...,M'_m)$. Simultaneous similarity means that there is an invertible matrix $S$ over $\bbF$ such that $M_i \cdot S = S \cdot M'_i$ for all $i \in [m]$. That is to say, there exists an isomorphism between the underlying vector spaces that maps each linear transformation $M_i$ to the corresponding $M'_i$ that operates on the isomorphic space. 
A structure $\AA$ satisfies $\Qq_f^{t,\ell}I(\bar{x})$ if $I(\AA)$ is a structure that encodes a tuple $(M_1,...,M_m)$ of matrices and satisfies $f(M_1,...,M_m) \geq t$.

Fragments of LA yield interesting equivalence relations between structures, which are approximations of isomorphism. The fragments that are studied in the literature (e.g.\ in \cite{limitationsOfIM, lichter2023separating, limitationsOfIM}) are parametrized by $k \in \bbN$ and $Q \subseteq \bbP$. The logic $\LA{k}{Q}$ is the $k$-variable fragment of LA
that only uses linear-algebraic operators over finite fields of characteristic $p \in Q$.
The equivalence relation induced by $\LA{k}{Q}$ is called \emph{invertible-map equivalence}. We write $\AA \ImEquiv{k}{Q} \BB$ if the two structures satisfy exactly the same $\LA{k}{Q}$-sentences. Invertible-map equivalence of two given structures can be tested in polynomial time \cite{limitationsOfIM}.

The logic $\LA{k}{Q}$ is at least as expressive as $\C{k}$ because the quantifier $\exists^{\geq i} x \phi(x)$ can be simulated with the rank operator~\cite{limitationsOfIM}: We have $\AA \models \exists^{\geq i} x \phi(x)$ if and only if the diagonal matrix which has a $1$-entry at exactly those positions $(a,a) \in A^2$ such that $\AA \models \phi(a)$ has rank at least~$i$. This works irrespective of which primes are in $Q$. Hence, the relation $\ImEquiv{k}{Q}$ is at least as fine as $\equivC{k}$; in fact, it is strictly finer because there exist generalised CFI-structures that are $\equivC{k}$-equivalent but distinguishable in rank logic \cite{rankLogicDiscovery}.

Invertible-map equivalence also has a characterization in terms of a Spoiler-Duplicator game called the \emph{invertible-map game}~\cite{DawarHolm12}.
We follow the exposition in \cite{lichter2023separating}.
Let $Q \subseteq \mathbb{P}$ and $k \in \bbN$. The IM-game $\Mm^{k,Q}$ is played on two structures $\AA$ and $\BB$. There are $k$ pairs of pebbles labelled with $1,\dots, k$.
A position in the game is a pair $\bar{a},\bar{b}$ of tuples $\bar{a} \in A^m$ and $\bar{b} \in B^m$ for some $m \leq k$.
In position $\bar{a},\bar{b}$ corresponding pebbles,
that is, pebbles with the same label, are placed on $a_i$ and $b_i$ for every $i \in [\ell]$.
In the initial position, all pebbles are placed beside the structures.
If $|A| \neq |B|$, then Spoiler wins immediately.
Otherwise, a round of the game is played as follows:
\begin{enumerate}
	\item Spoiler chooses a prime $p \in Q$ and a number $\ell$ satisfying $2\ell \leq k$. He picks up $2\ell$ pebbles from $\AA$  and the corresponding pebbles (with the same labels) from $\BB$.
	\item Duplicator picks a partition $\Pp$ of $A^\ell \times A^\ell$ and another one $\Pp'$ of $B^\ell \times B^\ell$ such that $|\Pp| = |\Pp'|$. Furthermore, she picks a bijection $f \colon \Pp \to \Pp'$ 
	and an invertible $(A^\ell \times B^\ell)$-matrix $S$ over $\bbF_p$ such that 
	$\chi^P = S \cdot \chi^{f(P)} \cdot S^{-1}$ for every $P \in \Pp$. Here, $\chi^P$ denotes the characteristic matrix of $P$, which has a $1$-entry at position $(\bar{u},\bar{v})$ if and only if $\bar{u}\bar{v} \in P$.
	\item Spoiler chooses a block $P \in \Pp$, a tuple $\bar{u} \in P$, and a tuple $\bar{v} \in f(P)$. Then for each $i \in [2\ell]$, he places one of the pebbles picked up from $\AA$ on~$u_i$ and the corresponding one picked up from~$\BB$ on $v_i$.	
\end{enumerate}	
After a round, Spoiler wins the game if the pebbles do not define a partial isomorphism or if Duplicator was not able to respond with a matrix satisfying the condition above.
Note that this condition states that the characteristic matrices of the blocks
are simultaneously similar.

\begin{lemma}
	Let $k \in \bbN$, $Q \subseteq \bbP$, $\AA$ and $\BB$ be structures,
	$\bar{a} \in A^k$, and $\bar{b} \in B^k$.
	Then $(\AA, \bar{a}) \ImEquiv{k}{Q} (\AA, \bar{b})$
	if and only if Duplicator has a winning strategy in the invertible-map game
	$\Mm^{k,Q}$ on $\AA$ and $\BB$ in position $\bar{a},\bar{b}$.
\end{lemma}
The lemma follows from a combination of~\cite{DawarHolm12, limitationsOfIM},
in which the game is also parametrised by the dimension $2\ell$ of the interpretations. In~\cite{DawarHolm12}, only finite sets of primes are considered because the logics considered there are not infinitary. The arguments straight-forwardly apply to arbitrary sets of primes.

\subparagraph{Homomorphism Indistinguishability.}
Let $F$ and $G$ be graphs. A \emph{homomorphism} $\psi$ from $F$ to $G$ is a map $\psi \colon V(F) \to V(G)$ such that $\psi(u)\psi(v) \in E(G)$ for every edge $uv \in E(F)$. We write $\hom(F, G)$ for the number of homomorphisms from $F$ to $G$.
Homomorphism counts induce equivalence relations on graphs: Let $\Ff$ be a class of graphs. Two graphs $G$ and $H$ are \emph{homomorphism indistinguishable} over $\mathcal{F}$, denoted by $G \homEquiv{\Ff} H$, if for every $F \in \Ff$, it holds $\hom(F,G) = \hom(F,H)$.
An equivalence relation $\approx$ between graphs is a \emph{homomorphism indistinguishability relation} if there exists a graph class $\mathcal{F}$ such that $\approx$ and $\equiv_{\mathcal{F}}$ coincide.

In this article, we call two graphs \emph{quantum isomorphic} if they are homomorphism indistinguishable over all planar graphs.
The term was introduced in \cite{atserias_quantum_2019} as a quantum information theoretic notion. The titular result of \cite{mancinska_quantum_2020} asserts that it is the same as homomorphism indistinguishability over all planar graphs. We do not need the original quantum based definition here.

\section{Homomorphisms to CFI-Like Graphs over Arbitrary Abelian Groups}
	\label{sec:roberson}
	
	Roberson~\cite{roberson2022oddomorphisms} studied homomorphisms to CFI-like graphs constructed over $\bbZ_2$.
	This variant of CFI graphs was introduced by Fürer~\cite{Furer2001}.
	Neuen and Schweitzer~\cite{NeuenSchweitzer17} generalised the more classical CFI construction from $\bbZ_2$ to arbitrary finite abelian groups.
	We combine both constructions and generalize the CFI construction from~\cite{Furer2001,roberson2022oddomorphisms}
	to arbitrary finite abelian groups.
	We fix such a group~$\Gamma$ throughout this section and write its operation as addition.
	
	For a graph $G$ and a vertex $v \in V(G)$, write $E(v) \coloneqq \{e \in E(G) \mid v \in e\}$ for the set of edges incident to $v$.
	We consider vectors $U \in \Gamma^X$ for e.g.\ $X = V(G)$.
	For $x \in X$, write $U(x) \in \Gamma$ for the $x$-th entry of $U$.
 	We write $\sum U$ for $\sum_{x \in X} U(x)$.
	If convenient, we denote by $x$ also the vector in $\Gamma^X$ with entry $1$ at the $x$-th position and $0$ everywhere else.
	
	\begin{definition}
		\label{def:robersonCFI}
		
		A \emph{base graph} is a connected graph.
		Let $G$ be a base graph and $U \in \Gamma^{V(G)}$.
	For every vertex $u$ of $G$, we define
	$$V_u := \setcond*{(u, S)}{S \in \Gamma^{E(u)}, \sum S = U(u)}.$$
	The \emph{CFI graph} $\CFI[\Gamma, G, U]$ over the finite abelian group $\Gamma$ and
	the base graph $G$ has vertex set $\bigcup_{u \in V(G)} V_u$
	and edge set
	\[\setcond[\big]{\{(u,S),(v,T)\}}{(u,S) \in V_u, (v,T) \in V_v, uv \in E(G), S(uv) + T(uv) = 0}.\]
	We say that the vertices in $V_u$ have \emph{origin} $u$.
	\end{definition}
	
	\begin{lemma}[restate=lemIso, name = ]
		\label{lem:iso}
		Let $G$ be a base graph and $U, U' \in \Gamma^{V(G)}$. If $\sum U = \sum U'$, then $\CFI[\Gamma, G,U] \cong \CFI[\Gamma, G,U']$.
	\end{lemma}
	The proof of \cref{lem:iso} uses well-known arguments for CFI graphs (see Appendix~\ref{app:roberson}).
For a graph $G$, and $U \in \Gamma^{V(G)}$, consider the \emph{projection map} $\rho \colon \CFI[\Gamma, G,U] \to G$ sending $(v, S)$ to~$v$. Clearly, $\rho$ is a homomorphism. For a graph $F$ and $\psi \colon F \to G$, define
	\[
	\Hom_\psi(F, \CFI[\Gamma, G,U]) \coloneqq \setcond[\big]{\phi \in \Hom(F, \CFI[\Gamma, G,U])}{\rho \circ \phi = \psi}.
	\]
	The sets $\Hom_\psi(F, \CFI[\Gamma, G,U])$ for all $\psi \colon F \to G$ partition the set $\Hom(F, \CFI[\Gamma, G,U])$ of homomorphisms $F \to \CFI[\Gamma, G,U]$.
	Write $\hom_\psi(F, \CFI[\Gamma, G,U])$ for the cardinality of $\Hom_\psi(F, \CFI[\Gamma, G,U])$.
	
	\begin{lemma} \label{lem:bijection}
		Let $F$ be a graph and $G$ be a base graph.
		Let $U \in \Gamma^{V(G)}$ and fix $\psi \in \Hom(F, G)$.
		Consider the system of equations $\mathsf{Hom}(F, G, U, \psi)$ with variables $x_e^a$ for all $a \in V(F)$ and $e \in E(\psi(a))$ and equations
		\begin{align}
			\sum_{e \in E(\psi(a))} x_e^a &= U(\psi(a)) && \text{for all } a \in V(F), \label{sys1}\\
			x_e^a + x_e^b &= 0 && \text{for all } ab \in E(F) \text{ and } e = \psi(ab) \in E(G).\label{sys2}
		\end{align}
		Then the number of solutions to $\mathsf{Hom}(F, G, U, \psi)$ over $\Gamma$ is $\hom_\psi(F, \CFI[\Gamma, G,U])$.
	\end{lemma}
	\begin{proof}
		The proof is by giving a bijection between the solution set and  $\Hom_\psi(F, \CFI[\Gamma, G,U])$.
		Let $x = (x_e^a)_{a \in V(F), e \in E(\psi(a))}$ be a solution to $\mathsf{Hom}(F, G, U, \psi)$. Define a homomorphism $\phi_x \in \Hom_\psi(F, \CFI[\Gamma, G,U])$ via 
		$\phi_x(a) \coloneqq \big(\psi(a), (x^a_e)_{e \in E(\psi(a))}\big)$.
		\Cref{sys1} guarantees that this is indeed a map from the vertices of $F$ to the ones of $\CFI[\Gamma, G,U]$.
		If $a$ and $b$ are adjacent in $F$, then so are $\psi(a)$ and $\psi(b)$ in $G$.
		Furthermore, $x_{\psi(ab)}^a + x_{\psi(ab)}^b = 0$ by \cref{sys2}.
		Hence, $\phi_x(a)$ and $\phi_x(b)$ are adjacent in $\CFI[\Gamma, G,U]$.
		
		It is easy to see that this construction is injective, i.e., if $\phi_x = \phi_y$, then $x = y$.
		For surjectivity, let $\phi \in \Hom_\psi(F, \CFI[\Gamma, G,U])$.
		For every $a \in V(F)$ and $e \in E(\psi(a))$, define $x_e^a$ as the second component of $\phi(a)$, i.e.\@ $x_e^a \coloneqq S_a(e)$ where $\phi(a) = (\psi(a), S_a)$. Clearly, $x = (x_e^a)$ is such that $\phi_x = \phi$. The fact that~$x$ satisfies \cref{sys1,sys2} is easily verified.
	\end{proof}

	\begin{theorem} \label{thm:equations}
		Let $G$ be a base graph and $U \in \Gamma^{V(G)}$. 
		Let $\psi\in \Hom(F, G)$ for some graph~$F$.
		\begin{enumerate}
			\item Then $\hom_\psi(F, \CFI[\Gamma, G,0]) > 0$.
			\item If $\mathsf{Hom}(F, G, U , \psi)$ has a solution, then $\hom_\psi(F, \CFI[\Gamma, G,0]) = \hom_\psi(F, \CFI[\Gamma, G,U])$.
			\item If $\mathsf{Hom}(F, G, U, \psi)$ has no solution,  then $\hom_\psi(F, \CFI[\Gamma, G,U]) = 0$.
		\end{enumerate}
	\end{theorem}
	\begin{proof}
		The system $\mathsf{Hom}(F, G, U, \psi)$ can be compressed into a matrix equation as follows:
		For $\psi \in \Hom(F, G)$
		and $P \coloneqq \{(a,e) \mid a \in V(F), e \in E(\psi(a))\}$, let $A^\psi \in \Gamma^{V(F) \times P}$ and $B^\psi \in \Gamma^{E(F) \times P}$  be the matrices such that
		\begin{equation} \label{eq:def-a-b}
			A^\psi_{b, (a,e)} = \delta_{b = a} \quad \text{and} \quad
			B^\psi_{bc, (a,e)} = \delta_{a \in \{b, c\} \land e = \psi(bc) }.
		\end{equation}
		Then \cref{sys1,sys2} are equivalent to
		\begin{equation} \label{realsys}
			\left( \begin{matrix}
				A^\psi \\
				B^\psi
			\end{matrix} \right) x
			= \left( \begin{matrix}
				U \circ \psi \\
				0
			\end{matrix} \right).
		\end{equation}
		If $U = 0$, then this system always has a solution, namely $x = 0$.
		In particular, by \cref{lem:bijection}, $\Hom_\psi(F, \CFI[\Gamma, G,0]) \neq \emptyset$.
Given \cref{lem:bijection}, it remains to give a bijection between the sets of solutions to $\left( \begin{smallmatrix}
			A^\psi \\
			B^\psi
		\end{smallmatrix} \right) x
		= \left( \begin{smallmatrix}
			0 \\
			0
		\end{smallmatrix} \right)$  and the set of solutions to $\left( \begin{smallmatrix}
			A^\psi \\
			B^\psi
		\end{smallmatrix} \right) x
		= \left( \begin{smallmatrix}
			U \circ \psi \\
			0
		\end{smallmatrix} \right)$. Provided with a solution $x^*$ to the latter system, $x \mapsto x + x^*$ can be taken to be this bijection.
	\end{proof}
	
	\begin{corollary} \label{cor:iso-hom}
		Let $G$ be a base graph and $U \in \Gamma^{V(G)}$. 
		Then the following are equivalent:
		\begin{enumerate}
			\item $\sum U = 0$,\label{it1}
			\item $\CFI[\Gamma, G,U] \cong \CFI[\Gamma, G,0]$,\label{it2}
			\item $\hom(G, \CFI[\Gamma, G,U]) = \hom(G, \CFI[\Gamma, G,0])$,\label{it3}
			\item $\hom_{\id}(G, \CFI[\Gamma, G,U]) = \hom_{\id}(G, \CFI[\Gamma, G,0])$, where $\id$ is the identity map on $G$.\label{it4}
		\end{enumerate}
	\end{corollary}
	\begin{proof}
		The fact that \cref{it1} implies \cref{it2} follows from \cref{lem:iso}.
		It is immediate that \cref{it2} implies \cref{it3}.
		The fact that \cref{it3} implies \cref{it4} follows from \cref{thm:equations}.
		It thus remains to prove that \cref{it4} implies \cref{it1}.
		
		By \cref{thm:equations}, let $x$ be a solution to \cref{realsys} for $\psi = \id \colon G \to G$. Then,
		\[
		\sum_{a \in V(G)} U(a) \overset{\eqref{sys1}}{=} \sum_{a \in V(G)}\sum_{e \in E(a)} x_e^a = \sum_{e = ab\in E(G)} x_e^a + x_e^b \overset{\eqref{sys2}}{=} 0.
		\]
		Hence, \cref{it1} holds.
	\end{proof}
Thus, if $G$ is planar, then it witnesses quantum non-isomorphism of the CFI graphs. 
	
	\begin{corollary}
		\label{cor:planar-base-implies-non-quantum-isomorphic}
		If $G$ is a planar base graph and $\sum U \neq 0$, then $\CFI[\Gamma, G,0]$ and $\CFI[\Gamma, G,U]$ are not quantum isomorphic.
	\end{corollary}

	\section{Invertible-Map Equivalence and Homomorphism Indistinguishability}
	\label{sec:inv-map-hom-indistinguish}
	
	In this section we prove that, for every $k \geq 6$, the invertible-map equivalence
	$\ImEquiv{k}{\bbP}$ over the set of all primes is not a homomorphism indistinguishability relation.
	The proof idea is the following:
	Using techniques from~\cite{lichter2023separating},
	we will construct, for every $k$, a planar base graph~$G$ 
	such that we obtain non-isomorphic but $\ImEquiv{k}{\bbP}$-equivalent generalised CFI graphs over~$G$
	and~$\bbZ_{2^i}$ for some~$i$.
	By \cref{cor:planar-base-implies-non-quantum-isomorphic},
	the two CFI graphs are not quantum isomorphic.
	Exploiting~\cite{seppelt2023logical}, we will see that this implies that~$\ImEquiv{k}{\bbP}$ is not a homomorphism-indistinguishability relation.
	
	\begin{lemma}
		\label{lem:homRelationRefinesQuantumIsomorphism}
		Let $k \geq 6$.
		If $\ImEquiv{k}{\bbP}$ (over graphs) is a homomorphism indistinguishability relation, then all $\ImEquiv{k}{\bbP}$-equivalent graphs are quantum isomorphic.
	\end{lemma}	
	\begin{proof}
		For every (self-complementary) logic $\Ll$,
		the following holds~\cite[Theorem 22]{seppelt2023logical}:
		If $\Ll$-equivalence is a homomorphism indistinguishability relation,
		and if, for every $\ell \in \bbN$, there are $\Cc^{\ell}$-equivalent
		but not $\Ll$-equivalent graphs $H$ and $H'$,
		then all $\Ll$-equivalent graphs are quantum isomorphic.
Here, $\Ll$ is $\LA{k}{\mathbb{P}}$.
		We show that for every $\ell \in \bbN$,
		there are $\Cc^{\ell}$-equivalent 
		but not $\LA{k}{\mathbb{P}}$-equivalent graphs $H$ and $H'$.
		Let $\ell \in \bbN$.
		It is well-known~\cite{caifurimm92}
		that there is a base graph~$G$ such that the two non-isomorphic  CFI graphs~$H$ and~$H'$ over~$\bbZ_2$ and~$G$, using the classical CFI construction (which we have not presented in this paper),
		are $\Cc^{\ell}$\nobreakdash-equivalent.
		However, the CFI graphs~$H$ and~$H'$ are not equivalent in rank logic~\cite{rankLogicDiscovery}.
		The interpretation defining the distinguishing matrices is actually one-dimensional and requires $6$ variables~\cite{Holm2011}.
		Thus,~$H$ and~$H'$ are not $\LA{k}{\mathbb{P}}$-equivalent.
	\end{proof}	
For now, assume the following lemma, which we will prove in the end of this section.
	\begin{lemma}
		\label{lem:planar-base-graph-im-equivalent}
		For every $k \in \bbN$, there is a planar base graph $G$
		and an $i \in \bbN$ such that, for all $U,U' \in \bbZ_{2^i}^{V(G)}$
		satisfying $\sum U = \sum U' + 2^{i-1}$, we have
		$\CFI[\bbZ_{2^i}, G, U] \ImEquiv{k}{\bbP}\CFI[\bbZ_{2^i}, G, U']$.
	\end{lemma}

	\begin{proof}[Proof of \cref{thm:main}]
		Let $k\geq 6$.
		By \cref{lem:planar-base-graph-im-equivalent},
		there is a planar base graph $G$ and an $i \in \bbN$
		such that $\CFI[\bbZ_{2^i}, G, 0] \ImEquiv{k}{\bbP}\CFI[\bbZ_{2^i}, G, U]$ for some $U \in \bbZ_{2^i}$ with $\sum U= 2^{i-1}$.
		These two CFI graphs are not quantum isomorphic
		by \cref{cor:planar-base-implies-non-quantum-isomorphic}.
		Hence, the invertible-map equivalence $\ImEquiv{k}{\bbP}$ is not a
		homomorphism indistinguishability relation by 	\cref{lem:homRelationRefinesQuantumIsomorphism}.
	\end{proof}
	Because the interpretation in the proof of \cref{lem:homRelationRefinesQuantumIsomorphism}
	is one-dimensional,
	the result of \cref{thm:main} also hold
	for equivalence in the fragment of $k$-variable linear-algebraic logic
	that is restricted to one-dimensional interpretations.
		
	It remains to prove \cref{lem:planar-base-graph-im-equivalent}.	
	Without the planarity requirement,
	non-isomorphic but $\ImEquiv{k}{\mathbb{P}}$-equivalent generalised CFI structures were
	constructed in~\cite{lichter2023separating}.
	By a careful analysis of the proof,
	the construction can be adapted to certain planar base graphs,
	which we will show now.
	However, we first have to extend our CFI graphs by additional relations.
	An ordered graph is a pair $(G,\leq)$
	of a graph $G$ and a total order $\leq$ on $V(G)$.
	If $G$ is an ordered graph, we denote its vertex set, its edge set, and its order by $V(G)$, $E(G)$, and $\leq^G$, respectively.
	\begin{definition}
		Let $i$ be a positive integer, $G$ be an ordered base graph,
		and $U \in \bbZ_{2^i}^{V(G)}$.
		We define the CFI structure $\CFI^*[\bbZ_{2^i}, G, U]$
		on the same vertex set as $\CFI[\bbZ_{2^i}, G, U]$,
		that is, on $\bigcup_{u \in V(G)} V_u$ (recall \cref{def:robersonCFI}).
		We first define a total preorder $\preceq$ on the vertices:
		$(u,S) \preceq (v,T)$ if and only if $u \leq^G v$.
 		For every $uv \in E(G)$, we define the following relations:
		\begin{align*}
			N_{u,v} &:= \setcond[\big]{
				((u,S),(u,T)) \in V_u^2} {S(uv)=T(uv)},\\
			C_{u,v} &:= \setcond[\big]{((u,S),(u,T)) \in V_u^2}{S(uv)+1=T(uv)}.
		\end{align*}
		Finally, we add for every $j \in \bbZ_{2^i}$ the following relation:
		\begin{align*}
			I_j := \setcond[\big]{\{(u,S),(v,T)\}}{(u,S) \in V_u, (v,T) \in V_v, uv \in E(G), S(uv) + T(uv) = j}.
		\end{align*}
	\end{definition}
	The structure $\CFI^*[\bbZ_{2^i}, G, U]$ can be seen as a vertex-coloured and edge-coloured
	directed graph.
	The preorder assigns colours to vertices, where vertices obtain the same color exactly if they have the same origin.
	The other relations colour edges by the set of relations in which they are contained.
	Note that $I_0$ coincides with the edge relation of the CFI graph $\CFI[\bbZ_{2^i}, G, U]$.
	The additional relations are, apart from the preorder, already implicit in  $\CFI[\bbZ_{2^i}, G, U]$ and are made explicit to ensure definability of certain properties in logics.

	Non-isomorphic but $\ImEquiv{k}{\mathbb{P}}$-equivalent CFI graphs
	were constructed using
	a class of regular base graphs,
	in which the degree, the girth, and the vertex-connectivity are simultaneously unbounded~\cite{lichter2023separating}.
	We will show that it suffices that the graph only satisfies these  properties ``locally''.
	The \emph{$r$-ball around a vertex} $v \in V$
	is the set of vertices with distance at most~$r$ to~$v$.
	
	\begin{definition}
	Let $G$ be a base graph and $r,d,g,c\in \bbN$.
	We say that $G$ is \emph{$(r, d,g,c)$-nice} if there is some vertex $w \in V(G)$
	such that the $r$-ball $W$ around $w$ satisfies the following:
	\begin{enumerate}
		\item Every vertex in $W$ has degree at least~$d$.
		\item Every cycle in $G$ containing a vertex of $W$ as length at least $g$.
		\item For every set $V' \subseteq V(G)$ of size at most $c$, all vertices in $W\setminus V'$ are contained in the same connected component of $G -V'$.
		\item For every set $V' \subseteq V(G)$ of size $c' \leq c$,
		there is at most one connected component of $G-V'$
		that is not an induced subgraph of a grid of height $c'$.
	\end{enumerate}
	\end{definition}

	\begin{lemma}
		\label{lem:nice-planar-graphs}
		For every $n \in \bbN$, there is a planar graph $G$ that is $(n,2n,2n,n)$-nice.
	\end{lemma}
	\begin{proof}
		Let $n \in \bbN$ be arbitrary but fixed.
		We start with a complete $2n$-ary tree (with fixed root~$w$) of depth~$4n$.
		For every $i \geq 1$, the $i$-th level of the tree
		consists of $(2n)(2n-1)^{i-1}$ vertices.
		In particular, the tree has $(2n)(2n-1)^{4n-1}$ leaves.
		Next, we attach a grid of height~$2n$ and width $(2n)(2n-1)^{4n-1}$ to the tree as follows:
		The $i$-th leaf from the left (according to the usual drawing of a tree in the plane) is identified with the $i$-th vertex of the grid in the first row.
		Denote this graph by~$G$. 
		It is easy to see that~$G$ is planar.
		We prove that~$G$ is $(n,2n,2n,n)$-nice, which is witnessed by the root~$w$.
		Let~$W$ be the~$n$-ball around~$w$,
		that is, the set of vertices whose level is at most $n+1$ in the tree.
		By construction, every vertex in~$W$ has degree~$2n$ and every cycle, in which a vertex of~$W$ is contained, has length at least~$2n$ because the tree has depth~$4n$.
		
		For every vertex $u \in W$, there are at least~$2n$ paths from $u$ into the grid that are disjoint apart from~$u$.
		Let $V'\subset V(G)$ be a set of at most~$n$ vertices.
		We show that all vertices in $W\setminus V'$ are connected in $G - V'$.
		Let $u,v \in W\setminus V'$.
		If there is a path from~$u$ to~$v$ only using vertices of the tree, we are done.
		Otherwise, there are at most~$n$ paths disjoint apart from~$u$ respectively~$v$ into the grid (because there were~$2n$ such paths for $u$ respectively $v$ before removing~$n$ vertices).
		Let~$V_u$ and~$V_v$ be the sets of endpoints of these paths,
		i.e., sets of size at least~$n$ of vertices in the first row of the grid.
		Because there is no path between~$u$ and~$v$ in the tree, at most $n-1$ vertices of the grid are removed in $G-V'$ (we count the leaves of the tree as vertices of the grid).
		By removing $n-1$ vertices from a grid of height~$2n$ (and larger width)
		it is not possible to separate the sets~$V_u$ and~$V_v$
		because they are of size at least~$n$ each.
		Hence, some vertex of~$V_u$ is connected to some vertex of~$V_v$
		in $G-V'$ and thus~$u$ and~$v$ are connected in $G-V'$.
		
		We finally show that at most one connected component of $G-V'$ 
		is not an induced subgraph of a grid of height at most $|V'|$.
		First, we claim that all vertices of the tree are in the same connected component of $G-V'$ (again, we count the leaves as vertices of the grid).
		One easily sees that the argument above actually works for all vertices of the tree
		because for all vertices of the tree there are~$2n$ disjoint paths into the grid.	
		So there is a component containing all vertices of the tree
		and some vertices of the grid.
		Second, because the grid has height and length greater than $n$,
		by removing $|V'| \leq n$ vertices from~$G$ we can only ``cut out'' holes or corners of the grid.
		This means that the component containing the tree vertices also contains all grid vertices apart from the holes and corners cut out.
		Each of them contains at most $|V'|$ vertices per column and thus all these holes and corners are induced subgraphs of a grid of height $|V'|$.
	\end{proof}
We now analyse properties of CFI structures over nice base graphs.
	The following proofs assume that the reader is familiar with the CFI construction.
	For more details we refer for example to~\cite{caifurimm92,Furer2001,GradelPakusa19,lichter2023separating}.
For some number $c \in \bbN$,
	a \emph{$c$-orbit} of a structure~$\AA$
	is a maximal set of $c$-tuples of~$\AA$
	that are all related by an automorphism of~$\AA$.
	That is, $\bar{x},\bar{y} \in A^c$ are in the same
	orbit if and only if there is an automorphism~$\phi$ of~$\AA$
	such that $\phi(\bar{x}) = \bar{y}$.
	The set of $c$-orbits is a partition of~$A^c$.
	
	We often need isomorphisms of a particular kind between generalised CFI structures.
	We have seen in \cref{lem:iso}
	that two CFI graphs $\CFI[\bbZ_{2^i},G,U]$ and $\CFI[\bbZ_{2^i},G,U']$ over some base graph $G$ are isomorphic if and only if $\sum U = \sum U'$.
	The same reasoning applies to the CFI structures
	$\CFI^*[\bbZ_{2^i},G,U]$ and $\CFI^*[\bbZ_{2^i},G,U']$ (see also \cite{lichter2023separating}).
	Let $p = u_1,\dots,u_m$ be a path in $G$ and $j \in \bbZ_{2^i}$.
	Now we can construct an isomorphism $\phi$ between
	$\CFI^*[\bbZ_{2^i},G,U]$ and $\CFI^*[\bbZ_{2^i},G,U-ju_1+ju_m]$
	(where $jv$ denotes the vector in $V(G)^{\bbZ_{2^i}}$
	that has entry $j$ at position $v$ and is zero otherwise)
	such that $\phi$ is the identity map on all vertices
	whose origin is not contained in $p$.
	This isomorphism can be composed out of the maps constructed in \cref{lem:iso} by following the path $p$.
	We call such isomorphisms \emph{path-isomorphisms}.
	If $p$ is a closed cycle, then the associated path-isomorphism is an automorphism of the structure, which we call \emph{cycle-automorphism}.

	\begin{lemma}
		\label{lem:nice-implies-homogeneous}
		Let $i \in \bbN$, $G$ be an $(r, d,g,c)$-nice ordered base graph, and $U \in \bbZ_{2^i}^{V(G)}$.
		Then two tuples of length $c' \leq c$ of $\CFI^*[\bbZ_{2^i},G, U]$ are $\Cc^{3c'}$-equivalent if and only if they are in the same $c'$-orbit.
	\end{lemma}
	\begin{proof}
		
	We start with the following special case:
	\begin{claim}[restate = niceImpliesHomogeneousStep, name = ]
		\label{cln:nice-implies-homogeneous-step}
		Let $\bar{a} = \bar{\gamma}x$ and $\bar{b} = \bar{\gamma}y$
		be tuples of length $c' \leq c$ of $\CFI^*[\bbZ_{2^i},G, U]$.
		If $\bar{a} $ and $\bar{b}$ are $\Cc^{3c'}$-equivalent,
		then $\bar{a}$ and $\bar{b}$ are in the same $c'$-orbit. 
	\end{claim}
	\begin{proof}[Proof Sketch]
		The vertices $x$ and $y$ must have the same origin $v$,
		so let $x = (v,S)$ and $y = (v,T)$ for some $S,T \in \bbZ_{2^i}^{E(v)}$.
		To construct an automorphism $\pi$ that pointwise fixes $\bar{\gamma}$
		and maps $x$ to $y$,
		we have to shift the edges $F:=\setcond{e \in E(v)} {S(e) \neq T(e)}$.
		Let $B$ be the set of all origins of vertices in $\bar{\gamma}$.
		Let $\Pp$ be the partition of $F$ according to the connected components of $G-B-\{v\}$
		into which the edges in $F$ lead.
		Such an automorphism $\pi$ exists if and only if
		every $P \in \Pp$ satisfies $\sum_{e\in P} S(e) -T(e) = 0$.
		Suppose this is not the case.
		At least two parts of $\Pp$ do not satisfy the condition, since $\sum S = \sum T$.
		Because $G$ is nice, the corresponding connected component
		of at least one of the parts is an induced subgraph of a grid of height~$c'$.
		Because non-isomorphic CFI graphs over grids of height~$c'$
		are not $\Cc^{3c'}$-equivalent~\cite{Furer2001}, 
		the tuples $\bar{a} $ and $\bar{b}$ are not $\Cc^{3c'}$-equivalent,
		which is a contradiction.
		For the full proof see \cref{app:inv-map-hom-indistinguish}.
	\end{proof}
To prove the lemma, first note that if two tuples are in the same orbit,
	then they are equivalent in every logic.
	So it remains to prove the other direction.
	We show by induction on the length~$c'$ of the tuples~$\bar{a}$ and~$\bar{b}$
	that if~$\bar{a}$ and~$\bar{b}$ are $\Cc^{3c'}$-equivalent,
	then they are in the same $c'$-orbit,
	i.e., there is an automorphism of $\CFI^*[\bbZ_{2^i},G,U]$ that maps~$\bar{a}$ to~$\bar{b}$.
	
	For $c'=1$, the result follows from Claim~\ref{cln:nice-implies-homogeneous-step} using~$\gamma$ as the empty tuple.
	For the inductive step, assume $\bar{a} = \bar{a}'x$ and $\bar{b} = \bar{b}'y$ are $\Cc^{3(c+1)'}$\nobreakdash-equivalent.
	Then~$\bar{a}'$ and~$\bar{b}'$ are $\Cc^{3c'}$\nobreakdash-equivalent.
	By induction, there exists an automorphism~$\pi'$ such that $\pi'(\bar{a}') = \bar{b}'$.
	Then the tuples~$\pi'(\bar{a})$ and~$\bar{b}$ agree on all entries except potentially the last one.
	They are $\Cc^{3(c+1)'}$\nobreakdash-equivalent because logical formulas do not distinguish between tuples in the same orbit.
	By Claim~\ref{cln:nice-implies-homogeneous-step},
	there is an automorphism $\pi$ such that $\pi(\pi'(\bar{a})) = \bar{b}$.
	So~$\bar{a}$ and~$\bar{b}$ are in the same orbit.
	\end{proof}
For a graph $G$, we call two sets $V,W\subseteq V(G)$
	adjacent if there are $v \in V$ and $w\in W$
	such that~$v$ and~$w$ are adjacent in~$G$.
	\begin{lemma}[restate = orbitsIndependentNice, name = ]
		\label{lem:orbits-independent-nice}
		Let $i \in \bbN$, $G$ be an $(r,d,g,c)$-nice ordered base graph witnessed by a vertex $w \in V(G)$, and let $U \in \bbZ_{2^i}^{V(G)}$.
		Furthermore, let $\phi$ be an automorphism of $\CFI^*[\bbZ_{2^i},G, U]$.
		If $\bar{x}$, $\bar{y}$, and $\bar{z}$ are tuples of $\CFI^*[\bbZ_{2^i},G, U]$
		such that
		\begin{enumerate}
			\item $|\bar{x}\bar{y}\bar{z}| \leq c$,
			\item the sets of all origins of vertices in $\bar{x}$, $\bar{y}$, and $\bar{z}$, respectively, are pairwise not adjacent in~$G$, and
			\item all origins of vertices in $\bar{x}$ and $\bar{y}$ are contained in the $(r-1)$-ball around $w$,
		\end{enumerate}
		then $\bar{x}\bar{y}\bar{z}$,  $\phi(\bar{x})\bar{y}\bar{z}$, and $\bar{x}\phi(\bar{y})\bar{z}$ are in the same orbit of $\CFI^*[\bbZ_{2^i},G, U]$.
	\end{lemma}
	The proof of \cref{lem:orbits-independent-nice}
	makes use of  standard arguments for CFI graphs and cycle-automorphisms.
	Such cycles can always be found for~$\bar{x}$ and~$\bar{y}$
	because removing all origins of vertices in~$\bar{x}$ and~$\bar{y}$
	does not disconnect~$G$ because~$G$ is nice (see \cref{app:inv-map-hom-indistinguish} for details).

	\begin{lemma}
		\label{lem:nice-implies-2equiv}
		For every $k \in \bbN$, there are $r,d,g,c,i \in \bbN$ such that,
		for every $(r,d,g,c)$\nobreakdash-nice ordered base graph $G$
		and every $U,U' \in \bbZ_{2^i}^V$ such that $\sum U = \sum U' + 2^{i-1}$,
		we have $\CFI^*[\bbZ_{2^i}, G, U] \ImEquiv{k}{\{2\}} \CFI^*[\bbZ_{2^i}, G, U']$.
	\end{lemma}
	\begin{proof}
		The proof is based on a close inspection of the proof in~\cite{lichter2023separating}:
		For every $2m \leq k$,
		base graphs of degree at least $d(m,k-2m)$,
		girth at least $g(m,k-2m)$,
		and vertex-connectivity at least $c(m,k-2m)$ are considered
		(for the definitions of~$d$,~$g$, and~$c$, see~\cite{lichter2023separating}).
		Of particular interest is the $r(m,k-2m)$-ball around some vertex,
		which we will see later.
		The CFI graphs are constructed over $\bbZ_{2^i}$, for some $i(m,k-2m) \in \bbN$.
		Define
		$d = d(k) := \max_{2m \leq k} d(m,k-2m)$
		and define $g= g(k)$, $c=c(k)$, $r=r(k)$, and $i=i(k)$ analogously.
		
		Assume $G$ is a $(2(k+2)r, d, g, c)$-nice and ordered base graph
		and let $u\in V(G)$ be a vertex witnessing this.
		We call the $4(k+2)r$-ball around $u$ the \emph{nice region} of $G$.
		Let $U,U' \in \bbZ_{2^i}^V$ with $\sum U = \sum U' + 2^{i-1}$
		and consider
		$\AA := \CFI^*[\bbZ_{2^i}, G, U]$ and $\BB := \CFI^*[\bbZ_{2^i}, G, U']$.
		To~prove $\AA \equiv^{\IM}_{k,\{2\}} \BB$,
		we show that Duplicator wins the characteristic $2$ IM game with $k$-pebbles $\Mm^{k,\{2\}}$ played on~$\AA$ and~$\BB$.	
		Duplicator maintains as invariant that
		in position $\bar{v}, \bar{v}'$,
		there is an isomorphism $\phi \colon \BB \to \BB'$
		where $\BB' := \CFI[\bbZ_{2^i}, G, U'']$ for some $U'' \in \bbZ_{2^i}^V$
		such that
		\begin{enumerate}
			\item $\phi(\bar{v}') = \bar{v}$,
			\item there is only a single vertex $w \in V$ such that $U(w) \neq U''(w)$ that we call \emph{twisted}, and
			\item the $(r+1)$-ball around $w$ is contained in the nice region and
			does not contain the origin of a vertex in $\bar{v}$.
		\end{enumerate}
		Clearly, the invariant holds initially.
		So assume that the invariant holds by the inductive hypothesis
		and that it is Spoiler's turn.
		W.l.o.g., we can assume to play on $\AA$ and $\BB'$
		in position $\bar{v},\bar{v}$.
		Spoiler chooses an arity $2m \leq k$
		and picks up~$2m$ pebbles from~$\AA$ and the corresponding ones (with the same labels) from~$\BB'$.
		Duplicator picks the $2m$-orbit partition~$\Pp$ of $(\AA,\bar{v})$,
		and the $2m$-orbit partition~$\Pp'$ of $(\BB,\bar{v})$.
		We construct a suitable bijection $\Pp \to \Pp'$ using the techniques of~\cite{lichter2023separating}.
If $G$ were regular with degree at least $d(m,k-2m)$, of girth at least $g(m, k-2m)$, and of vertex-connectivity at least $c(m, k-2m)$,
		then there would indeed be a similarity matrix as required by the game~\cite{lichter2023separating}.
		One crucial property of base graphs with vertex-connectivity strictly larger than $k$ is the following:
		Let $\bar{x}\bar{y}$ be a tuple of $\AA$ of length at most $k$
	 	such that the set of all origins of vertices in $\bar{x}$ is not adjacent to the
		same set for $\bar{y}$.
		In this case, automorphisms can be applied independently, that is,
		if $\phi$ is an automorphism, then $\bar{x}\bar{y}$ 
		is in the same orbit as $\phi(\bar{x})\bar{y}$, $\bar{x}\phi(\bar{y})$,
		and $\phi(\bar{x}\bar{y})$.
		The construction of the similarity matrix in~\cite{lichter2023separating} heavily depends on this fact.
		However, non-trivial automorphisms are only applied to such parts of tuples,
		for which all entries are contained in the $r(m,k-2m)$-ball around the twisted vertex (called the ``active region'' in~\cite{lichter2023separating}).
		This still holds for the $(2(k+2)r,d,g,c)$-nice base graph~$G$,
		if the $r(m,k-2m)$-ball around the twisted vertex $w$ is contained in the nice region:
		Let $\bar{x}\bar{y}\bar{z}$ be a tuple of vertices of~$\AA$ of length at most~$k$
		such that the sets of all origins of vertices of~$\bar{x}$,~$\bar{y}$, and respectively~$\bar{z}$ are pairwise not adjacent and the sets of all origins of vertices of~$\bar{x}$ and~$\bar{y}$
		are contained within the $r$-ball around~$w$.
		Then automorphisms can be applied independently in the sense above (Lemma~\ref{lem:orbits-independent-nice}).
		Hence, the same construction of the similarity matrix of~\cite{lichter2023separating}
		can also be applied here.
		All arguments requiring large girth and degree only consider vertices in the ``active region'', for which we also have long cycles and large degree in the nice region.
		
		Spoiler pebbles a $2m$-tuple in some block $P \in \Pp$
		and a $2m$-tuple in $f(P) \in \Pp'$
		resulting in the position~$\bar{v}''$ and~$\bar{v}'''$.
		By the properties of the similarity matrix
		and the bijection from~\cite{lichter2023separating},
		the pebbles define a partial isomorphism, and 
		there is an isomorphism $\psi \colon \AA \to \BB''$
		such that $\psi(\bar{v}'') = \bar{v}'''$
		and there is only a single twisted vertex between~$\AA$ and~$\BB''$.
		That is, Conditions~1 and~2 of the invariant are satisfied.
		
		To satisfy Condition~3, we use a path-isomorphism to move the twist to a vertex
		which has distance at least~$r$ to all origins of vertices in~$\bar{v}$
		as follows.
		Because~$G$ is nice, we can move the twist to all vertices in the nice region (because removing the origins of pebbled vertices does not separate the nice region).
		Because at most~$k$ vertices in the nice region are pebbled,
		the vertices in at most~$k$ many $r$-balls have distance less than~$r$ to 
		the all origins of vertices in~$\bar{v}''$.
		Since the nice region is a $4(k+2)r$-ball around~$u$,
		there is a vertex in the nice region whose $(r+1)$-ball is not pebbled and contained in the nice region.
		We move the twist to such a vertex.
		Duplicator maintains the invariant and wins the invertible-map game.
	\end{proof}
	
	\begin{lemma}
		\label{lem:planar-base-graph-im-equivalent-structures}
		For every $k \in \bbN$, there is a planar ordered base graph $G$
		and an $i \in \bbN$ such that, for all $U,U' \in \bbZ_{2^i}^{V(G)}$
		with $\sum U = \sum U' + 2^{i-1}$, we have
		$\CFI^*[\bbZ_{2^i}, G, U] \ImEquiv{k}{\bbP}\CFI^*[\bbZ_{2^i}, G, U']$.
	\end{lemma}
	\begin{proof}
		Let $k \in \bbN$ be arbitrary.
		Let $r,d,g,c$, and $i$ be the constants given by \cref{lem:nice-implies-2equiv} for $k' := 3k+1$
		and let $\ell := \max\{r,d,g,c\}$.
		By \cref{lem:nice-planar-graphs}, there is a planar graph $G=(V,E)$
		that is $(\ell,2\ell,2\ell,\ell)$-nice.
		One easily sees that $G$ is also $(r,d,g,c)$-nice.
		Hence, \[\CFI^*[\bbZ_{2^i}, G, U] \equiv^{\IM}_{3k+1,\{2\}} \CFI^*[\bbZ_{2^i}, G, U']\] by \cref{lem:nice-implies-2equiv} for all $U,U'\in \bbZ_{2^i}^V$
		with $\sum U = \sum U' + 2^{i-1}$.
		By \cref{lem:nice-implies-homogeneous},
		the $k'$\nobreakdash-orbits of these CFI structures
		are $\Cc^{3k'}$-definable
		and hence the class of CFI structures over $(\ell,2\ell,2\ell,\ell)$-nice and ordered base graphs is homogeneous in the sense of~\cite{limitationsOfIM}. From~\cite{limitationsOfIM} it follows that
		\[\CFI^*[\bbZ_{2^i}, G, U] \equiv^{\IM}_{3k+1,\mathbb{P} \setminus \{2\}} \CFI^*[\bbZ_{2^i}, G, U'].\]
		To show that these two equivalences 
		imply 
		\[\CFI^*[\bbZ_{2^i}, G, U] \equiv^{\IM}_{k,\mathbb{P}} \CFI^*[\bbZ_{2^i}, G, U'],\]
		we use the arguments from \cite[Lemma 10]{dawarGraedelLichter}.
		The authors prove for $k'' = k +2$ the following:
		If the $k$-orbits of two structures $H$ and $H'$ are definable in $\Cc^{k''}$ and for two sets of primes $P$ and $Q$ we have
		$H \equiv^{\IM}_{k''+1,P} H'$  and $H\equiv^{\IM}_{k''+1,Q} H'$,
		then $H \equiv^{\IM}_{k, P \cup Q} H'$.
		The same argument also applies for $k'' = 3k$
		and the claim of the lemma is proven.
	\end{proof}
	
	\begin{proof}[Proof of \cref{lem:planar-base-graph-im-equivalent}]
		Because $\CFI[\bbZ_{2^i}, G, U]$
		is up to renaming relation symbols
		a reduct of $\CFI^*[\bbZ_{2^i}, G, U]$
		(only the relation $I_0$ is kept),
		$\CFI^*[\bbZ_{2^i}, G, U] \equiv^{\IM}_{k,\mathbb{P}} \CFI^*[\bbZ_{2^i}, G, U']$ (Lemma~\ref{lem:planar-base-graph-im-equivalent-structures}) implies
		$\CFI[\bbZ_{2^i}, G, U] \equiv^{\IM}_{k,\mathbb{P}} \CFI[\bbZ_{2^i}, G, U']$.
	\end{proof}
	
\section{Comonads}
	In \cite{abramsky_pebbling_2017}, comonads on the category of relational structures were introduced which capture equivalences over certain fragments of first-order logic.
	For example, the \emph{pebbling comonad}~$\mathbb{T}_k$ has the property that two structures $\AA$ and $\BB$ satisfy the same sentences over $k$-variable first-order logic with counting quantifiers if and only if they are isomorphic in the co-Kleisli-category of~$\mathbb{T}_k$.
We refer the reader to \cite{dawar_lovasz-type_2021} and the previously mentioned references for formal definitions. 
	The following Lovász-type theorem for comonads allows us to derive \cref{thm:no-comonad} from \cref{thm:main}:
	
	\begin{theorem}[\cite{dawar_lovasz-type_2021,reggio_polyadic_2021}] \label{thm:dawar}
		Let $\mathbb{C}$ be a finite-rank comonad on the category of (not necessarily finite) graphs. Then there exists a graph class $\mathcal{F}$ such that two finite graphs  are isomorphic in the co-Kleisli category of $\mathbb{C}$ if and only if they are homomorphism indistinguishable over~$\mathcal{F}$.
	\end{theorem}
	For a definition of finite rank, see \cite[Definition~B.2]{reggio_polyadic_2021}.
	Less generally, one may think of a finite-rank comonad as a comonad which sends finite structures to finite structures. 
	Note that \cref{thm:no-comonad} does not rule out that invertible-map equivalence can be characterised comondically in a different way, i.e., not as co-Kleisli isomorphism but via a more involved construction.

\section{Modular Homomorphism Indistinguishability}

	In this section, we consider homomorphism indistinguishability modulo integers $n \in \mathbb{N}$.
	For a graph class $\mathcal{F}$, two graphs $G$ and $H$ are said to be \emph{homomorphism indistinguishable over~$\mathcal{F}$ modulo~$n$}, in symbols $G \homEquivMod{\Ff}{n} H$, if
	$\hom(F, G) \equiv \hom(F, H) \mod n$ for every $F \in\Ff$.
	We write $G \homEquivMod{\Ff}{N} H$ for $N \subseteq \mathbb{N}$ if $G \homEquivMod{\Ff}{n} H$ for every $n \in N$.
	
	In contrary to the classical result of Lovász~\cite{lovasz_operations_1967} asserting that two graphs are homomorphism indistinguishable over all graphs if and only if they are isomorphic,
	homomorphism counts modulo a prime $p$ do not suffice to determine a graph up to isomorphism.
	In \cite{faben_complexity_2015}, homomorphism indistinguishability over all graphs modulo~$p$ was characterised as follows:
	For a graph $G$ with automorphism $\sigma$, write $G^\sigma$ for the subgraph of $G$ induced by the fixed points of $\sigma$.
	Write $G \to_p G'$ for two graphs $G$ and $G'$ if there is an automorphism $\sigma$ of $G$ of order~$p$ such that $G^\sigma \cong G'$
	and write $G \to^*_p H$ if there is a sequence of graphs $G_1, \dots, G_n$ such that $G \to_p G_1 \to_p G_2 \to_p \dots \to_p G_n \to_p H$.
	By \cite[Theorem~3.7]{faben_complexity_2015}, for every graph $G$ and prime~$p$,
	there is a graph $G^*_p$, unique up to isomorphism, such that $G^*_p$ has no automorphisms of order~$p$, and $G \to^*_p G^*_p$. Furthermore, by \cite[Theorem~3.4]{faben_complexity_2015}, $G$ and $G^*_p$ are homomorphism indistinguishable over all graphs modulo~$p$.
	
	\begin{theorem}[{\cite[Lemma~3.10]{faben_complexity_2015}}]
		\label{thm:homCountingModPIsIsomorphismOfpRigidSubgraphs}
		Let $p$ be a prime. Two graphs $G$ and $H$ are homomorphism indistinguishable over all graphs modulo~$p$ if and only if $G^*_p$ and $H^*_p$ are isomorphic.
	\end{theorem}
In general, modular homomorphism indistinguishability relations are rather oblivious to striking differences between graphs:
	
	\begin{example} \label{ex:cliquecoclique}
		For $n \in \mathbb{N}$, the one-vertex graph $K_1$ and
		the coclique $\overline{K_{n+1}}$ are homomorphism indistinguishable over all graphs modulo~$n$.
	\end{example}
	\begin{proof}
		If $F$ is an edge-less graph, then $\hom(F, K_1) = 1 \equiv (n+1)^{|V(F)|} = \hom(F, \overline{K_{n+1}}) \mod n$.
		If otheriwse $F$ contains an edge, then $\hom(F, K_1) = 0 = \hom(F, \overline{K_{n+1}})$.
	\end{proof}
Before we move to modular homomorphism indistinguishability characterisations for certain logic fragments, we clarify the relationship between the various notions introduced so far:

	\begin{lemma} \label{lem:modhomind} \label{cor:homCountingModInfinitePrimes}
		Let $\mathcal{F}$ and $\mathcal{M}$ be graph classes.
		Let $N \subseteq \mathbb{N}$ and $n \in \mathbb{N}$.
		\begin{enumerate}
			\item If $N$ is infinite, then $\equiv^N_{\mathcal{F}}$ and $\equiv_{\mathcal{F}}$ coincide.
			\item If $N$ is finite and $m$ is the least common multiple of the numbers in $N$, then  $\equiv^N_{\mathcal{F}}$ and $\equiv^m_{\mathcal{F}}$ coincide.
			\item If $\equiv_{\mathcal{F}}$ and $\equiv^n_{\mathcal{M}}$ coincide, then $\mathcal{F} = \emptyset$, i.e., all graphs are $\equiv_{\mathcal{F}}$-equivalent.
		\end{enumerate}
	\end{lemma}
	\begin{proof}
		For the first claim, let $G$ and $H$ be graphs and $F \in \mathcal{F}$.
		Since $N$ is infinite, there exists $n \in N$ greater than $|V(G)|^{|V(F)|}$ and $|V(H)|^{|V(F)|}$.
		Then $\hom(F, G) \equiv \hom(F, H) \mod n$ implies that $\hom(F, G) = \hom(F, H)$.
		
		For the second claim, first observe that $G \equiv^m_{\mathcal{F}} H$ entails $G \equiv^N_{\mathcal{F}} H$ since all $n \in N$ divide~$m$. Conversely,
		for a prime $p$ write $\nu(p)$ for the greatest integer $k \geq 0$ such that there is an $n\in N$ that is divisible by $p^k$.
		Then $m = \prod_{p \in \bbP} p^{\nu(p)}$, where the product ranges over all primes.
		Hence, if $\hom(F, G) \equiv \hom(F, H) \mod n$ for all $n \in N$, then $\hom(F, G) \equiv \hom(F, H) \mod p^{\nu(p)}$ for all primes $p$ appearing as divisors of elements in $N$, i.e., $\nu(p) > 0$.
		Hence, by the Chinese Remainder Theorem, also $\hom(F, G) \equiv \hom(F, H) \mod m$.
		
For the third claim, suppose $G \in \mathcal{F}$ towards a contradiction. 
		Write $\ell$ for the maximum integer such that $p^\ell$ divides $n$ for some prime $p$.
		Write $\phi \colon \mathbb{N} \to \mathbb{N}$ for Euler's totient function and $G^{\times k}$ for the $k$-th categorical power of the graph $G$, cf.\@~\cite[p.\@~40]{lovasz_large_2012}.
		\begin{claim}\label{cl:phi}The graphs $G^{\times (\phi(n)+ \ell)}$ and $G^{\times \ell}$
			are homomorphism indistinguishable over all graphs modulo $n$.
		\end{claim}
		\begin{claimproof}
			We show that $a^{\ell}(a^{\phi(n)} - 1) \equiv 0 \mod n$ for every $a \in \mathbb{N}$.
			By the Chinese Remainder Theorem, 
			writing $n = \prod p_i^{\ell_i}$ as product of prime factors,
			it suffices to show this equality modulo $p_i^{\ell_i}$ for every $i$.
			By Euler's Theorem, $a^{\phi(p_i^{\ell_i})} \equiv 1 \mod p_i^{\ell_i}$ if $a$ and $p_i$ are coprime. Since $\phi(n) = \prod \phi(p_i^{\ell_i})$, also $a^{\phi(n)} \equiv 1 \mod p_i^{\ell_i}$.
			If $p_i$ divides $a$, then $a^\ell \equiv 0 \mod p_i^{\ell_i}$ as $\ell_i \leq \ell$.
			Finally,
			for every graph $F$, $\hom(F, G^{\times (\phi(n)+ \ell)}) = \hom(F, G)^{\phi(n)+ \ell} \equiv \hom(F, G)^{\ell} \mod n$ by~\cite[(5.30)]{lovasz_large_2012}.
		\end{claimproof}
		Let $c \in \mathbb{N}$ be greater than the chromatic number of $G$, in particular satisfying that $\hom(G, K_c) > 1$.
		By \cite[(5.30)]{lovasz_large_2012}, we then have that ${\hom(G, (G \times K_c)^{\times (\phi(n)+\ell)}) \neq \hom(G, (G \times K_c)^{\times \ell}})$ because $\hom(G, G \times K_c) > 1$ and $\phi(n) \geq 1$.
		However, $(G \times K_c)^{\times (\phi(n)+\ell)} \equiv_{\mathcal{M}}^n (G \times K_c)^{\times \ell}$ by \cref{cl:phi} contradicting that $\equiv_{\mathcal{F}}$ and $\equiv^N_{\mathcal{M}}$ coincide.
	\end{proof}
\Cref{lem:modhomind} shows that non-trivial modular homomorphism indistinguishability relations cannot be expressed by (non-modular) homomorphism indistinguishability relations. Furthermore, considering sets of moduli does not yield more relations. We may restrict our attention to homomorphism indistinguishability relations modulo some not necessarily prime $n \in \mathbb{N}$. In the remainder of this section, we give an example and a non-example of a logic whose equivalence can be characterised as modular homomorphism indistinguishability relation.

We have seen already that the relation $\equiv_{k,\mathbb{P}}^{\IM}$ is not a homomorphism indistinguishability relation over any graph class. But since $\equiv_{k,\mathbb{P}}^{\IM}$ is a relation based on linear algebra over finite fields, it might a priori be that it can be characterised as a homomorphism indistinguishability relation \emph{modulo a prime}. 
	This can be ruled out, at least in the following sense:

	\begin{theorem}
		\label{thm:IMnoHomRelationModP}
		Let $k \geq 2$ and $P$ be a set of primes. 
		Then there exists no graph class~$\mathcal{F}$ and no $n \in \mathbb{N}$ such that $\equiv_{k,P}^{\IM}$ and $\equiv^n_{\mathcal{F}}$ coincide.
	\end{theorem}
	\begin{proof}
		Towards a contradiction, suppose that $\equiv_{k,P}^{\IM}$ and $\equiv^n_{\mathcal{F}}$ coincide for some graph class~$\mathcal{F}$ and some $n \in \mathbb{N}$.
		Consider the clique $K_{1}$ and the coclique $\overline{K_{n+1}}$, which are homomorphism indistinguishable over all graphs modulo~$b$ by~\cref{ex:cliquecoclique}, and the FO-sentence ${\phi \coloneqq \exists x_1 \exists x_2.\ x_1 \neq x_2}$.
		Clearly, $K_{1} \not\models \phi$ while $\overline{K_{n+1}} \models \phi$.
		Hence, $2$-variable FO distinguishes the two graphs and $K_{1} \not\equiv_{k,P}^{\IM} \overline{K_{n+1}}$.
	\end{proof}	
\label{sec:treewidth}
	By extending techniques of \cite{dvorak_recognizing_2010}, we prove that homomorphism indistinguishability over graphs of bounded treewidth counted modulo a prime characterises  equivalence in first-order logic with modular counting quantifiers.
	For a definition of treewidth, see \cite{bodlaender_partial_1998} or \cref{app:treewidth}.
	Let $p$ be a prime.
	Let $\Cmod{p}$ denote the set of formulas inductively defined as follows:
	\begin{itemize}
		\item for variables $x$ and $y$, the formulas $x = y$ and $E(x,y)$ are in $\Cmod{p}$,
		\item if $\phi, \psi \in\Cmod{p}$, then $\neg \phi, \phi \land \psi, \phi \lor \psi \in \Cmod{p}$, and
		\item if $\phi \in \Cmod{p}$, $x$ is a variable, and $c \in \mathbb{F}_p$, then $\exists^c x.\ \phi$ is in $\Cmod{p}$.
	\end{itemize}
	The semantics is as expected, i.e., a structure $\AA$ satisfies a sentence $\exists^{c} x.\ \phi(x)$ if there exist $c \mod p$ distinct $a \in A$ such that $\AA \models \phi(a)$.
	Let $\CkMod{k+1}{p}$ denote the $(k+1)$-variable fragment of this logic.

	\begin{theorem} \label{thm:dvorak-mod-p}
		Let $p$ be a prime and $k \geq 0$.
		Two arbitrary graphs $G$ and $H$ are homomorphism indistinguishable over all graphs of treewidth at most $k$ modulo~$p$
		if and only if
		$G$ and $H$ are $\CkMod{k+1}{p}$-equivalent.
	\end{theorem}

	\section{Conclusion}
	
	We studied linear-algebraic logic, a logic stronger than first-order logic with counting, and proved that equivalence with respect to it can neither be characterised as a homomorphism indistinguishability relation, nor as co-Kleisli isomorphism for a finite-rank comonad. The latter answers an open question of {\'O}~Conghaile and Dawar~\cite{conghaile_game_2021} and shows a limitation of the game comonad programme for capturing logical equivalences. It would be desirable to understand more generally which properties are responsible for making a logic suitable for a homomorphism indistinguishability or game comonad characterisation. 
	We know that game comonads can be defined for FO with all Lindström quantifiers up to a fixed arity \cite{conghaile_game_2021} -- what we do not know is whether these are the only Lindström extensions of FO admitting such a characterisation. Other interesting classes of Lindström quantifiers to look at besides the linear-algebraic ones could be CSP quantifiers. The corresponding logic defined in \cite{hella_expressive_2023} comes with a fairly natural game characterising equivalence. Thus, one may ask whether this CSP logic admits a game comonad or if this can be ruled out with similar methods as in this paper. 
	The same question is also open for (bounded variable fragments of) counting monadic second order logic CMSO.
	In principle, our approach works for every extension of counting logic for which there exists a CFI-like lower bound construction that works over planar base graphs and with only one binary relation. It remains to devise such a construction for CSP logic and CMSO.
	
	A different topic, that we have merely touched upon, is homomorphism counting in prime fields. We have shown that the corresponding homomorphism indistinguishability relations do not characterise IM-equivalence. On the other hand, we stated an example of a logic that is captured by a modular homomorphism indistinguishability relation, namely FO with modulo counting quantifiers.
	A more comprehensive theory of modular homomorphism counting is yet to be developed. A particularly interesting question, which is not in the scope of this article, is whether the known connections between homomorphism counting and solutions to semidefinite/linear programs for graph isomorphism \cite{roberson_lasserre_2023} have a meaningful generalisation to prime fields.

	\bibliographystyle{plainurl}

	\newpage
	
	\appendix
	
	\section{Material Omitted in \Cref{sec:roberson}}
	\label{app:roberson}
		Let $\Gamma$ be an arbitrary finite abelian group.
		\lemIso*
		\begin{proof}
		Let $uv \in E(G)$.
		Denote the vertex set of $\CFI[\Gamma, G,U]$ respectively $\CFI[\Gamma, G,U']$ by $V_U$ and $V_{U'}$.
		First consider $U' \coloneqq U + u - v$ where $u$ and $v$ denote the vectors in $\Gamma^{V(G)}$ with one at the $u$-th and $v$-th component, respectively, and zero otherwise. Define
		the map $\phi \colon V_U \to V_{U'}$ by
		\[
		\phi((w, S)) \coloneqq \begin{cases}
			(u, S + uv), & \text{if } w = u, \\
			(v, S - uv), & \text{if } w = v, \\
			(w, S), & \text{otherwise.}
		\end{cases}
		\]
		where $uv$ denotes the vector in $\Gamma^{E(u)}$ in the first case or in $\Gamma^{E(v)}$ in the second case with one at the $uv$-th component and zero otherwise.
Observe that $\sum_{e \in E(v)} (S-uv)(e) = U(v) - 1 = U'(v)$ and analogously for $u$. Hence, $\phi$ is indeed a well-defined map to $V_{U'}$. Clearly, $\phi$ is a bijection.
Let $(x,S),(y, T) \in V_U$ be arbitrary vertices of $\CFI[\Gamma, G,U]$ and write $\phi(x,S) \eqqcolon (x, S')$ and $\phi(y, T) \eqqcolon (y, T')$. Then $S'(xy) + T'(xy) = S(xy) + T(xy)$.
		Hence, $(x,S)$ and $(y, T)$ are adjacent in $\CFI[\Gamma, G,U]$ if and only if they are adjacent in $\CFI[\Gamma, G,U']$.
		
		Since $G$ is connected, the maps constructed above can be composed to yield $\CFI[\Gamma, G,U] \cong \CFI[\Gamma, G,U + u - v]$ for every pair of vertices $u, v$. This yields $\CFI[\Gamma, G,U] \cong \CFI[\Gamma, G,U']$ as desired.
	\end{proof}
	
	\section{Material Omitted in Section~\ref{sec:inv-map-hom-indistinguish}}
	\label{app:inv-map-hom-indistinguish}
	
	Let $G$ be an ordered base graph that is $(r,d,g,c)$-nice,
	$i \in \bbN$, and $U \in \bbZ_{2^i}^{V(G)}$.
	\niceImpliesHomogeneousStep*
	\begin{claimproof}
		The vertices~$x$ and~$y$ must have the same origin $v \in V(G)$ (if they are not, then they can easily be distinguished using~$\preceq$ and~$3$ variables). 
		So let $x = (v,S)$ and  $y = (v, T)$ for some $S, T \in \bbZ_{2^i}^{N_G(v)}$.
		We need to find an automorphism~$\pi$ of $\CFI[\bbZ_{2^i},G,U]$ such that~$\pi$ fixes the tuple~$\bar{\gamma}$ pointwise and $\pi(x) = y$. In case that $x= y$, there is nothing to show, so assume $x \neq y$.
		
		First assume that some vertex $z$ in~$\bar{\gamma}$ also has origin~$v$.
		In this case,~$x$ and~$y$ can be distinguished by a $(c'+2)$-variable formula,
		using~$\gamma$ as parameters and the relations $C_{v,u}$:
		There must be a vertex $u$ such that the distance only via $C_{v,u}$-edges between $z$ and $x$ is different from the one between $z$ and $y$, which can be expressed using $2$ additional variables.
		This contradicts that~$\bar{a}$ and~$\bar{b}$ are $\Cc^{3c'}$\nobreakdash-equivalent.
		
		So assume that all vertices in~$\bar{\gamma}$ do not have origin~$v$.
		Let $F := \setcond{ e \in E(v)}{S(e) \neq T(e)}$ denote the incident edges of~$v$ that~$\pi$ has to \emph{shift}.
		Since $\sum S = \sum T$, we have $|F| > 1$.
		Let $B \subseteq V(G)$ be the set of all origins of vertices in $\bar{\gamma}$. 
		Because~$\pi$ must fix~$\bar{\gamma}$, the automorphism~$\pi$ cannot apply a shift to any edge that is incident to a vertex in~$B$.
		One can see that~$\pi$ with the required properties exists if and only if there exists a partition~$\Pp$ of~$F$ such that for every part $P \in \Pp$, $\sum_{e \in P} S(e) -T(e) = 0$ and the edges in~$P$ lead into the same connected component of $G - B - \{v\}$:
		
		If the condition holds, then we can construct~$\pi$ as follows: For each $P \in \Pp$,
		pick a vertex~$w$ in the component~$C$ of $G - B - \{v\}$ that is connected with the edges in~$P$.
		For each $e \in P$, fix a path in~$C$ that starts with the edge~$e$ and ends in~$w$.
		Then, for every $e \in P$,~$\pi$ shifts the edge~$e$ at vertex~$v$ by $T(e) - S(e)$. The shifts are propagated to~$w$ by a path-isomorphism along the respective path. The sum of the shifts at~$w$ will be $\sum_{e \in P} S(e) -T(e) = 0$.
		Hence, this defines an automorphism with the desired properties.
		It is not difficult to see that if such a partition~$\Pp$ of~$F$ does not exist, then there is no~$\pi$ that fixes~$\bar{\gamma}$ and maps~$x$ to~$y$.
		In this case, it is not possible to apply the required shifts to the edges in~$F$ in such a way that no edge in~$B$ is shifted and the shifts at each vertex sum up to zero.
		
		Thus, it remains to show that if the required partition of~$F$ does not exist, then~$\bar{\gamma}x$ and~$\bar{\gamma}y$ are not $\Cc^{3c'}$-equivalent.
		If such a partition does not exist, then $G - B - \{v\}$ must be disconnected.
		By the definition of niceness (since $|B| \leq c' \leq c$), all but at most one connected components of $G - B$ are induced subgraphs
		of a grid of height~$c'$.
		For at least two components the required shifts do not sum up to zero (because all together they sum up to zero).
		Hence, there is at least one component that is an induced subgraph of a grid of height~$c'$ -- call it~$X$ -- for which the shifts in the edges in $F$ leading from~$v$ into~$X$ do not sum up to zero.
		Let $X_B := X \cup \setcond{ u \in B}{ut \in E(G) \text{ for some } t \in X}$. Let $\AA = \CFI[\bbZ_{2^i},G,U]$
		and denote by~$\AA[X_B]$ the substructure of~$\AA$
		induced by all vertices whose origin is in~$X_B$.
		To distinguish $(\AA[X_B], \bar{\gamma}x)$
		from $(\AA[X_B], \bar{\gamma}y)$,
		we can actually distinguish $(\AA[X_B], \bar{\gamma}x)$
		from $(\AA'[X_B], \bar{\gamma}x)$,
		where $\AA' = \CFI[\bbZ_{2^i},G,U']$ for some $U' \in V(G)^{\bbZ_{2^i}}$
		such that $U(u) = U'(u)$ for all $u \in B$
		and $\sum_{u \in X} U(x)$ and $\sum_{u \in X} U'(x)$
		differ by the sum of shifts required into~$C$.
		This is the case
		because $(\AA'[X_B], \bar{\gamma}x)$ 
		is isomorphic to $(\AA[X_B], \bar{\gamma}y)$.
		Instead of distinguishing $(\AA[X_B], \bar{\gamma}x)$
		from $(\AA[X_B], \bar{\gamma}y)$
		in $\Cc^{3c'}$ by a formula with~$c'$ free variables
		interpreted as~$\bar{\gamma}y$,
		it essentially suffices to distinguish $\CFI[\bbZ_{2^i},G[X],U[X]]$
		from $\CFI[\bbZ_{2^i},G[X],U'[X]]$ in $\Cc^{2c'}$.
		Since $G[X]$ is an induced subgraph of a grid of height $c'$,
		it follows from $\cite{Furer2001}$
		that the two structures are not $\Cc^{c'+1}$-equivalent
		so in particular not $\Cc^{2c'}$-equivalent.
	\end{claimproof}
	
	\orbitsIndependentNice*
	\begin{proof}
		Let $O(\bar{x}) \subseteq V(G)$ be the set of all origins of vertices in~$\bar{x}$ and define $O(\bar{y})$ and $O(\bar{z})$ similarly.
		It is well-known for CFI graphs that every automorphism is composed of cycle-automorphisms.
		For our CFI structures over $\bbZ_{2^i}$ see e.g.~\cite{lichter2023separating}.
		It suffices to show that there are automorphisms $\psi$ and $\pi$
		such that $\psi(\bar{x}\bar{y}\bar{z}) = \phi(\bar{x})\bar{y}\bar{z}$
		and $\pi(\bar{x}\bar{y}\bar{z}) =\bar{x}\phi(\bar{y})\bar{z}$.
		We construct such an automorphism $\psi$ as follows:
		First decompose $\phi$ into cycle-automorphisms,
		that is, $\phi = \phi_1 \circ \cdots \circ \phi_m$,
		where all the~$\phi_i$ are cycle-automorphisms.
		First assume that the cycle corresponding to every~$\phi_i$ contains an origin of a vertex in~$\bar{x}$.
		Assume $i \in [m]$ and let $c_i = u_1,\dots, u_\ell,u_1$ be the  cycle corresponding to~$\phi_i$
		and let $j_1, \dots, j_n$ be all indices such that~$u_{j_i}$ is a vertex in~$\bar{x}$.
		Note that the vertices $u_{j_i+1}$ and $u_{j_i-1}$ for all $i \in [n]$ (indices wrapping around) are all contained in the $r$-ball around~$w$.
		That is, these vertices are contained in the same connected component of 
		$G - O(\bar{x}) - O(\bar{y}) - O(\bar{z})$ because~$G$ is nice.
		Hence, for every $i \in [n]$,
		there are (possibly empty) paths~$p_i$ from $u_{j_i+1}$ to $u_{j_{i+1}-1}$ that do not contain vertices of $O(\bar{x})$, $O(\bar{y})$,
		and $O(\bar{z})$ (indices again wrapping around).
		So we obtain a cycle $c_i' = u_{j_1} p_1 \cdots p_{n-1} u_{j_n} p_n u_{j_1}$
		(which possibly uses vertices multiple times)
		that does not contain vertices of $O(\bar{y})$ and $O(\bar{z})$
		and contains exactly the edges that are incident to vertices in $O(\bar{x})$ and contained in~$c_i$.
		Actually, these edges are used in the same direction in~$c_i'$ as in~$c_i$.
		Let~$\psi_i$ be the cycle-automorphism corresponding to~$c_i'$.
		Then $\psi_i(\bar{y}) = \bar{y}$ and $\psi_i(\bar{z}) = \bar{z}$
		because~$c_i'$ does not contain vertices of $O(\bar{y})$ and $O(\bar{z})$.
		Furthermore, we have $\psi_i(\bar{x}) = \phi_i(\bar{x})$
		because~$c_i'$ uses the same edges as~$c_i$ that are incident to $O(\bar{x})$ and in particular uses them in the same direction.
		Now let $\psi = \psi_1 \circ \cdots \circ \psi_m$.
		One easily sees that $\psi(\bar{x}\bar{y}\bar{z}) = \phi(\bar{x})\bar{y}\bar{z}$.
		
		If some~$\psi_i$ does not contain vertices of $O(\bar{x})$, it is just ignored in the construction.
		By analogous reasoning, we can construct the desired automorphism~$\pi$.
	\end{proof}
	
	\section{Material Omitted in \Cref{sec:treewidth}}
	\label{app:treewidth}
	
	Let $F$ be a graph. A \emph{tree decomposition} for $F$ is a tuple $(T, \beta)$ where $T$ is a tree and $\beta \colon V(T) \to 2^{V(F)}$ is such that
	\begin{enumerate}
		\item $\bigcup_{t \in V(T)} \beta(t) = V(F)$,
		\item for all $uv \in E(F)$, there exists $t \in V(T)$ such that $u, v \in \beta(t)$,
		\item for all $u \in V(F)$, the set of all $t \in V(T)$ such that $u \in \beta(t)$ induces a connected subgraph of $T$.
	\end{enumerate}
	The \emph{width} of $(T, \beta)$ is $\max_{t\in V(T)}|\beta(t)| -1$. The \emph{treewidth} of $F$ is the minimal width over all tree decomposition of $F$. See \cite{bodlaender_partial_1998} for further details on treewidth.
	
	Let $k \geq 1$.
	A \emph{$k$-labelled graph} is a tuple $\boldsymbol{F} = (F, \boldsymbol{u})$ where $F$ is a graph and $\boldsymbol{u} \in V(F)^k$. A homomorphism between $k$-labelled graphs $(F, \boldsymbol{u}) \to (G, \boldsymbol{w})$ is a homomorphism $h \colon F \to G$ such that $h(u_i) = w_i$ for all $i\in [k]$. We write $\hom(\boldsymbol{F}, \boldsymbol{G})$ for the number of homomorphisms between the labelled graphs $\boldsymbol{F}$ and $\boldsymbol{G}$.

	For two $k$-labelled graphs $\boldsymbol{F} = (F, \boldsymbol{u})$ and $\boldsymbol{K} = (K, \boldsymbol{v})$, define their \emph{gluing product} $\boldsymbol{F} \odot \boldsymbol{K}$ as the $k$-labelled graph whose underlying graph is obtained by taking the disjoint union of $F$ and $K$ and identifying $u_i$ and $v_i$ for all $i \in [k]$. It can be easily seen that $\hom(\boldsymbol{F} \odot \boldsymbol{K}, \boldsymbol{G}) = \hom(\boldsymbol{F}, \boldsymbol{G}) \hom(\boldsymbol{K}, \boldsymbol{G})$ for all $\boldsymbol{G}$. See \cite{mancinska_quantum_2020,grohe_homomorphism_2022} for further details.

	We consider a certain family of $(k+1)$-labelled graphs:
	\begin{definition} \label{def:twk}
		Let $\mathcal{TW}^k$ be the family of $(k+1)$-labelled graphs $\boldsymbol{F} = (F, \boldsymbol{u})$ such that $F$ admits a tree decomposition $(T, \beta)$ such that
		\begin{enumerate}
			\item there exists a bag $r \in V(T)$ such that $\beta(r) = \{u_1, \dots, u_{k+1}\}$,\label{it:twk1}
			\item if $|V(T)| \geq 2$, then $|\beta(t)| = k+1$ for all $t \in V(T)$ and $|\beta(s) \cap \beta(t)| = k$ for all $st \in E(T)$.\label{it:twk2}
		\end{enumerate}
	\end{definition}

	\begin{theorem}\label{thm:dvorak-mod-p-local}
		Let $p$ be a prime. Let $k \geq 1$.
		For all $(k+1)$-labelled graphs $\boldsymbol{G}$ and $\boldsymbol{H}$, the following are equivalent:
		\begin{enumerate}
			\item $\hom(\boldsymbol{F}, \boldsymbol{G}) \equiv \hom(\boldsymbol{F}, \boldsymbol{H}) \mod p$ for all $\boldsymbol{F} \in \mathcal{TW}^k$.
			\item For all formulae $\phi(x_1, \dots, x_{k+1}) \in \CkMod{k+1}{p}$, we have $\boldsymbol{G} \models \phi$ if and only if $\boldsymbol{H} \models \phi$.
		\end{enumerate}
	\end{theorem}
	Here, for $\boldsymbol{G} = (G, \boldsymbol{u})$, the expression $\boldsymbol{G} \models \phi$ indicates that $G \models \phi(u_1, \dots, u_{k+1})$.
	
	\begin{lemma} \label{lem:graph-to-formula}
		Let $p$ be a prime and $k \geq 1$.
		For every $\boldsymbol{F} \in \mathcal{TW}^k$ and every $m \in \mathbb{F}_p$, there exists a formula $\phi_m(x_1, \dots, x_{k+1}) \in \CkMod{k+1}{p}$ such that for every $(k+1)$-labelled graph $\boldsymbol{G}$, 
		\[
		\boldsymbol{G} \models \phi \iff \hom(\boldsymbol{F}, \boldsymbol{G}) \equiv m \mod p.
		\]
	\end{lemma}
	\begin{proof}
		Let $\boldsymbol{F} = (F, \boldsymbol{u})$ with tree decomposition $(T, \beta)$ be as in \cref{def:twk}.
		The proof is by induction on the size of $T$.
		
		If $|V(T)| = 1$, then all vertices of $F$ are labelled and $\hom(\boldsymbol{F},\boldsymbol{G}) \in \{0,1\}$ for every $\boldsymbol{G}$.
		If~$m$ is neither $0$ nor $1$ modulo~$p$, set $\phi_m$ to $\mathsf{false}$.
		If $m = 1$, set
		\[
		\phi_1 \coloneqq \bigwedge_{\substack{1 \leq i \neq j \leq k+1 \\ u_i = u_j}} (x_i = x_j) \land \bigwedge_{\substack{1 \leq i \neq j \leq k+1 \\ u_iu_j \in E(F)}} E(x_i,x_j).
		\]
		Then $\boldsymbol{G} \models \phi_1 \iff \hom(\boldsymbol{F},\boldsymbol{G}) = 1 \iff \hom(\boldsymbol{F},\boldsymbol{G}) \equiv 1 \mod p$, where the second equivalence holds since $\hom(\boldsymbol{F},\boldsymbol{G}) \in \{0,1\}$. Finally, set $\phi_0 \coloneqq \neg \phi_1$.
		
		If $|V(T)| \geq 2$, let $r \in V(T)$ denote the vertex from \cref{it:twk1} of \cref{def:twk}.
		First consider the case when $r$ has a single neighbour $s$ in $T$.
		Let $S$ denote the tree obtained from $T$ by deleting $r$.
		Write $F'$ for the subgraph of $F$ induced by $\bigcup_{t \in V(S)} \beta(t)$.
		By \cref{it:twk2} of \cref{def:twk}, one may find $\boldsymbol{v} \in V(F')^{k+1}$ such that $\beta(s) = \{v_1, \dots, v_{k+1}\}$ and $v_i = u_i$ for all $i \in [k+1] \setminus \{\ell\}$ for some $\ell \in [k+1]$.
		Then $\boldsymbol{F}' \coloneqq (F', \boldsymbol{v}) \in \mathcal{TW}^k$.
		Furthermore, let $A \coloneqq F[\beta(r)]$ and $\boldsymbol{A} \coloneqq (A, \boldsymbol{u}) \in \mathcal{TW}^k$. Then, writing $\boldsymbol{G} = (G, \boldsymbol{w})$,
		\[
		\hom(\boldsymbol{F}, \boldsymbol{G}) = \hom(\boldsymbol{A}, \boldsymbol{G}) \sum_{w \in V(G)} \hom(\boldsymbol{F}', (G, w_1 \dots w_{\ell-1} w w_{\ell+1}\dots w_{k+1} )).
		\]
		Let $\psi_m$ and $\chi_m$ denote the formulae constructed inductively for $\boldsymbol{A}$ and $\boldsymbol{F}'$ respectively. Let
		\[
		\phi_m \coloneqq \bigvee_{\substack{m', m'' \in \mathbb{F}_p, \\ m'm'' = m.}} \Big( \psi_{m'} \land 
		\bigvee_{\substack{c_1, \dots, c_p \in \mathbb{F}_p, \\ \sum_{i=1}^p i c_i = m''}} \bigwedge_{1 \leq i \leq p}   \exists^{c_i} x_\ell.\ \chi_{i} \Big). 		
		\]
		If $\boldsymbol{G} \models \phi_m$, then there exist $m', c_1, \dots, c_p \in \mathbb{F}_p$ such that $\hom(\boldsymbol{A}, \boldsymbol{G}) \equiv m' \mod p$, 
		\[ \left|\setcond[\big]{w \in V(G) }{ \hom(\boldsymbol{F}', (G, w_1 \dots w_{\ell-1} w w_{\ell+1}\dots w_{k+1} ) \equiv i \mod p }\right| \equiv c_i \mod p \] for all $i \in [p]$, and $m' \sum_{i=1}^p ic_i = m$.
		Hence, $\hom(\boldsymbol{F}, \boldsymbol{G}) \equiv m' \sum_{i=1}^p i c_i = m$.
		The converse is readily verified.
		
		It remains to consider the case when $r$ has multiple neighbours. 
		In this case, $\boldsymbol{F} = \boldsymbol{F}^1 \odot \dots \odot \boldsymbol{F}^r$ for some graphs $\boldsymbol{F}^1, \dots, \boldsymbol{F}^r \in \mathcal{TW}^k$ falling into the case considered above. Let $\phi^1_m, \dots, \phi^r_m$ denote the corresponding inductively constructed formulae. Set
		\[
		\phi_m \coloneqq \bigvee_{\substack{m_1, \dots, m_r \in \mathbb{F}_p, \\ m_1 \cdots m_r = m.}} \phi^1_{m_1} \land \dots \land \phi^r_{m_r}.
		\]
		Since $\hom(\boldsymbol{F}, \boldsymbol{G}) = \prod_{i=1}^r \hom(\boldsymbol{F}^i, \boldsymbol{G})$, this formula is as desired.
	\end{proof}

	Subsequently, we consider finite $\mathbb{F}_p$-linear combinations of graphs in $\mathcal{TW}^k$.
	For such a linear combination $\mathfrak{q} \coloneqq \sum \alpha_i \boldsymbol{F}^i$ with $\alpha_i \in \mathbb{F}_p$ and $\boldsymbol{F}^i \in \mathcal{TW}^k$, write $\hom(\mathfrak{q}, \boldsymbol{G}) \coloneqq \sum \alpha_i \hom(\boldsymbol{F}^i, \boldsymbol{G}) \in \mathbb{F}_p$.
	Write $\mathbb{F}_p\mathcal{TW}^k$ for the set of all such linear combinations.
	The gluing operation $\odot$ can be extended linearly to turn $\mathbb{F}_p\mathcal{TW}^k$ into an $\mathbb{F}_p$-algebra.
	Observe that $\hom(\mathfrak{q}_1 \odot \mathfrak{q}_2, \boldsymbol{G}) = \hom(\mathfrak{q}_1, \boldsymbol{G})\hom(\mathfrak{q}_1, \boldsymbol{G})$ for all $\mathfrak{q}_1, \mathfrak{q}_2 \in \mathbb{F}_p\mathcal{TW}^k$.
	
	\begin{lemma} \label{lem:interpolation}
		Let $p$ be a prime and $k \geq 1$.
		Let $\mathfrak{q} \in \mathbb{F}_p\mathcal{TW}^k$ and
		$X_1 \subseteq \mathbb{F}_p$.
		Then there exists $\mathfrak{r} \in \mathbb{F}_p\mathcal{TW}^k$ such that for all $(k+1)$-labelled graphs $\boldsymbol{G}$,
		\begin{itemize}
			\item if $\hom(\mathfrak{q}, \boldsymbol{G}) \not\in X_1$, then $\hom(\mathfrak{r}, \boldsymbol{G}) = 0$,
			\item if $\hom(\mathfrak{q}, \boldsymbol{G}) \in X_1$, then $\hom(\mathfrak{r}, \boldsymbol{G}) = 1$.
		\end{itemize}
	\end{lemma}
	\begin{proof}
		Consider the Lagrange polynomial $p(X) = \sum_{x \in X_1} \prod_{y \in \mathbb{F}_p \setminus \{x\}} \frac{X - y}{x-y} \in \mathbb{F}_p[X]$.
		Observe that $p(z) = 1$ if $z \in X_1$ and $p(z) = 0$ if $z \not\in X_1$.
		Define $\mathfrak{r} \coloneqq p(\mathfrak{q})$ via the $\mathbb{F}_p$-algebra structure of $\mathbb{F}_p\mathcal{TW}^k$.
		Then $\hom(\mathfrak{r}, \boldsymbol{G}) = p(\hom(\mathfrak{q}, \boldsymbol{G}))$, as desired.
	\end{proof}
	
	\begin{lemma} \label{lem:formala-to-graph}
		For every $\phi(x_1, \dots, x_{k+1}) \in \CkMod{k+1}{p}$, there exists a $\mathfrak{q} \in \mathbb{F}_p\mathcal{TW}^k$ such that  for all $(k+1)$-labelled $\boldsymbol{G}$,
		\begin{itemize}
			\item if $\boldsymbol{G} \not\models \phi$, then $\hom(\mathfrak{q}, \boldsymbol{G}) = 0$,
			\item if $\boldsymbol{G} \models \phi$, then $\hom(\mathfrak{q}, \boldsymbol{G}) = 1$.
		\end{itemize}
	\end{lemma}
	In this case, we say that $\mathfrak{q}$ \emph{models} $\phi$.
	\begin{proof}
		By induction on the structure of $\phi$.
		\begin{itemize}
			\item If $\phi = \mathsf{true}$, then the graph $\boldsymbol{I} = (I, (1,\dots, k+1)) \in \mathcal{TW}^k$ with $V(I) = [k+1]$ and $E(I) = \emptyset$ models $\phi$.
			\item If $\phi = (x_i = x_j)$ for $i < j$, then the graph $\boldsymbol{I}^{ij} = (I^{ij}, (1, \dots, j-1, i, j+1, \dots, k+1))  \in \mathcal{TW}^k$ with $V(I^{ij}) = [k+1] \setminus \{j\}$ and $E(I) = \emptyset$ models $\phi$.
			\item If $\phi = E(x_i, x_j)$ and $i \neq j$, then the graph $\boldsymbol{A}^{ij} = (A^{ij}, (1, \dots, k+1))  \in \mathcal{TW}^k$ with $V(A^{ij}) = [k+1]$ and $E(I) = \{ij\}$ models $\phi$.
			\item If $\phi = E(x_i, x_j)$ and $i = j$, then $0$ models $\phi$.
			\item If $\phi = \phi_1 \land \phi_2$, let $\mathfrak{t}_1, \mathfrak{t}_2 \in \mathbb{F}_p\mathcal{TW}^k$ denote the elements modelling $\phi_1$ and $\phi_2$ respectively. Their product $\mathfrak{t}_1 \odot \mathfrak{t}_2$ models $\phi$.
			\item If $\phi = \neg \phi_1$ and $\mathfrak{t}_1$ is as above, then $\boldsymbol{I} - \mathfrak{t}_1$ models $\phi$. 
			\item If $\phi = (\exists^{m}x_{\ell})\psi$, then let $\mathfrak{t} = \sum \alpha_i \boldsymbol{F}^i$ denote the element modelling $\psi$. 
			For every $\boldsymbol{F}^i = (F^i, \boldsymbol{v}^i)$, construct a graph $\boldsymbol{K}^i = (K^i, \boldsymbol{w}^i)$ by letting $V(K^i) \coloneqq V(F^i) \sqcup \{x\}$, $E(K^i) \coloneqq E(F^i)$, and $w^i_j \coloneqq v^i_j$ for all $j \in [k+1] \setminus \{\ell\}$ and $w^i_\ell \coloneqq x$.
			For a tree decomposition $(T_i, \beta_i)$ for $F^i$ as in \cref{def:twk} with vertex $r_i \in V(T_i)$ such that $\beta_i(r_i) = \{\boldsymbol{v}^i_1,\dots, \boldsymbol{v}^i_{k+1} \}$, distinguish cases:
			\begin{itemize}
				\item If $|V(T_i)| \geq 2$, define a tree $S_i$ by $V(S_i) \coloneqq V(T_i) \sqcup \{s_i\}$ and $E(S_i) \coloneqq E(T_i) \sqcup \{r_is_i\}$.
				Extend $\beta$ to a map defined on $S_i$ by letting $\beta(s_i) \coloneqq \{w^i_1, \dots, w^i_{k+1}\}$.
				\item If $|V(T_i)| = 1$ and $|V(K_i)| \leq k+1$, define a tree decomposition for $K_i$ with a single bag.
				\item If $|V(T_i)| = 1$ and $|V(K_i)| = k+2$, define a tree decomposition on the single edge tree with bags $\{w^i_1, \dots, w^i_{k+1}\}$ and $\{v^i_1, \dots, v^i_{k+1}\}$.
			\end{itemize}
			In any case, the tree decomposition is as in \cref{def:twk} and hence $\boldsymbol{K}^i \in \mathcal{TW}^k$. Observe that for $\boldsymbol{G} = (G, \boldsymbol{x})$,
			\[
			\hom(\boldsymbol{K}^i, \boldsymbol{G}) = \sum_{x \in V(G)} \hom(\boldsymbol{F}^i, (G, x_1 \dots x_{\ell-1} x x_{\ell+1}\dots x_{k+1} )).
			\]
			Let $\mathfrak{q} \coloneqq \sum \alpha_i \boldsymbol{K}^i \in \mathbb{F}_p\mathcal{TW}^k$.
			By induction, 
			\begin{align*}
				\hom(\mathfrak{q}, \boldsymbol{G}) 
				&\equiv \sum_{x \in V(G)} \hom(\mathfrak{t}, (G, x_1 \dots x_{\ell-1} x x_{\ell+1}\dots x_{k+1} )) \\
				&\equiv |\{x \in V(G) \mid (G, x_1 \dots x_{\ell-1} x x_{\ell+1}\dots x_{k+1} ) \models \psi\}| \mod p.
			\end{align*}
			The desired graph can now be easily constructed via \cref{lem:interpolation}. \qedhere
		\end{itemize} 
	\end{proof}
	
	\begin{proof}[Proof of \cref{thm:dvorak-mod-p-local}]
		The forward direction follows from \cref{lem:formala-to-graph}, the backward direction from \cref{lem:graph-to-formula}.
	\end{proof}
	
	\begin{proof}[Proof of \cref{thm:dvorak-mod-p}]
		First suppose that $G$ and $H$ satisfy the same $\CkMod{k+1}{p}$-sentences.
		Let~$F$ be a graph of treewidth at most $k$. By \cite[Lemma~8]{bodlaender_partial_1998}, there exists $\boldsymbol{u} \in V(F)^{k+1}$ such that $(F, \boldsymbol{u}) \in \mathcal{TW}^k$.
		Let $m \in [p]$.
		For every $0 \leq \ell \leq k+1$ and starting with $\ell = k+1$, we construct inductively a formula $\phi^\ell_m(x_1, \dots, x_\ell)$ with $\ell$ free variables such that
		\begin{equation} \label{eq:formula-length-induction}
			\hom((F, u_1 \dots u_\ell), (G, x_1 \dots x_\ell)) 
			\equiv m \mod p \iff (G, x_1 \dots x_\ell) \models \phi^\ell_m.
		\end{equation}
		By \cref{lem:graph-to-formula}, there exists a formula $\phi^{k+1}_m$ satisfying this condition.
		For $0 \leq \ell < k+1$, observe that
		\begin{equation}\label{eq:formula-length-induction2}
			\hom((F, u_1 \dots u_\ell), (G, x_1 \dots x_\ell)) 
			= \sum_{x \in V(G)} \hom((F, u_1 \dots u_{\ell+1}), (G, x_1 \dots x_\ell x)).
		\end{equation}
		Given $\phi^{\ell+1}_m$, define
		\[
		\phi^\ell_m \coloneqq \bigvee_{\substack{c_1,\dots, c_p \in \mathbb{F}_p \\ \sum i c_i = m}} \bigwedge_{i \in [p]} \exists^{c_i} x_{\ell+1}.\ \phi^{\ell+1}_{i}(x_1, \dots, x_{\ell+1}).
		\]
		Then this formula has $\ell$ free variables and satisfies \cref{eq:formula-length-induction}.
		Indeed, if \[\hom((F, u_1 \dots u_\ell), (G, x_1 \dots x_\ell)) 
		\equiv m \mod p,\] then, by \cref{eq:formula-length-induction2}, with
		\[
		c_i \coloneqq |\{x \in V(G) \mid \hom((F, u_1 \dots u_{\ell+1}), (G, x_1 \dots x_\ell x)) \equiv i \mod p\}|
		\]
		for $i \in [p]$, it holds that $\sum_{i \in [p]} i c_i \equiv m \mod p$ implying that $(G, x_1 \dots x_\ell) \models \phi^\ell_m$.
		Conversely, let $c_1, \dots, c_p \in \mathbb{F}_p$ be such that $\sum ic_i = m$ and $(G, x_1 \dots x_\ell) \models \exists^{c_i} x_{\ell+1}.\ \phi_i^{\ell+1}(x_1, \dots, x_{\ell+1})$ for all $i \in [p]$. By \cref{eq:formula-length-induction2}, $\hom((F, u_1 \dots u_\ell), (G, x_1 \dots x_\ell)) \equiv m \mod p$.
		
		Finally, we have that $\hom(F, G) \equiv m \mod p$ if and only if $G \models \phi^0_m$ for the sentence $\phi^0_m$.
		Hence, $G$ and $H$ are homomorphism indistinguishable over all graphs of treewidth at most $k$ modulo~$p$.
		
		Conversely, suppose that $G$ and $H$ are homomorphism indistinguishable over all graphs of treewidth at most $k$ modulo~$p$.
		Let $\phi$ be a $\CkMod{k+1}{p}$-sentence. We have to show that $G \models \phi$ if and only if $H \models \phi$.
		If $\phi$ is of the form $\neg \psi$ or $\psi \land \chi$ for some $\CkMod{k+1}{p}$-sentence $\psi$ and $\chi$,
		these sentences can be considered separately. Hence, it can be assumed that $\phi$ is of the form $\exists^c x_i.\ \psi(x_i)$ for some $i \in [p]$, $c \in \mathbb{F}_p$, and  a $\CkMod{k+1}{p}$-formula $\psi$ with one free variable.
		Let $\chi(x_1, \dots, x_{k+1}) \coloneqq \psi(x_1) \land \bigwedge_{j \neq i} (x_j = x_i)$. 
		Then $G \models \exists^c x_i.\ \psi(x_i)$ if and only if $G \models \exists^c x_1 \exists^1 x_2 \dots \exists^1 x_{k+1}.\ \chi(x_1, \dots, x_{k+1})$ and analogously for $H$.
		Let $\mathfrak{q} = \sum_{j \in J} \alpha_j \boldsymbol{F}^j \in \mathbb{F}_p\mathcal{TW}^k$ be as in \cref{lem:formala-to-graph} for $\chi$ with some finite index set $J$.
		Write $F^j$ for the unlabelled graph underlying $\boldsymbol{F}^j$.
		Then 
		\begin{align*}
			\left|\setcond[\big]{\boldsymbol{x} \in V(G)^{k+1} }{ (G, \boldsymbol{x}) \models \chi}\right|
			& \equiv \sum_{\boldsymbol{x} \in V(G)^{k+1}} \hom(\mathfrak{q}, (G, \boldsymbol{x})) \\
			& \equiv \sum_{j \in J} \alpha_j \hom(F^j, G) \\
			& \equiv \sum_{j \in J} \alpha_j \hom(F^j, H) \\
			& \equiv \sum_{\boldsymbol{y} \in V(H)^{k+1}} \hom(\mathfrak{q}, (H, \boldsymbol{y})) \\
			& \equiv \left|\setcond[\big]{\boldsymbol{y} \in V(H)^{k+1} }{ (H, \boldsymbol{y}) \models \chi}\right| \mod p.
		\end{align*}
		In particular, $G \models \phi$ if and only if $H \models \phi$.
	\end{proof}
\end{document}